\documentclass[12pt,a4paper]{article}

\usepackage{graphicx,natbib}
\usepackage{float,color}
\usepackage{amsmath,amsfonts,amsthm} % Math packages
\usepackage{algorithm}
\usepackage{algpseudocode}

\def\Y{\mathcal{Y}}

\def\ru{{\rm U}}

\def\Cov{{\rm Cov}}

\def\RR{{\mathbf{ R}}}

\numberwithin{equation}{section} % Number equations within sections (i.e. 1.1, 1.2, 2.1, 2.2 instead of 1, 2, 3, 4)
\numberwithin{figure}{section} % Number figures within sections (i.e. 1.1, 1.2, 2.1, 2.2 instead of 1, 2, 3, 4)
\numberwithin{table}{section} % Number tables within sections (i.e. 1.1, 1.2, 2.1, 2.2 instead of 1, 2, 3, 4)

\usepackage{array}
\usepackage{multirow}
\usepackage{url}
\usepackage{float}
\usepackage{hhline}
\usepackage{subcaption}
\usepackage{wrapfig}

\usepackage{fullpage}

\let\oldtabular\tabular
\renewcommand{\tabular}{\scriptsize\oldtabular}

\newtheorem{theorem}{Theorem}

\theoremstyle{remark}
\newtheorem*{remark}{Remark}

\title{\bf Adaptive Component-wise Multiple-Try Metropolis Sampling}
\author{\large Jinyoung Yang\thanks{Department of Statistical Sciences, University of Toronto, Toronto, Ontario, Canada M5S 3G3. Email: jinyoung.yang@mail.utoronto.edu}
\ , Evgeny Levi\thanks{Department of Statistical Sciences, University of Toronto, Toronto, Ontario, Canada M5S 3G3. Email: evgeny@utstat.utoronto.ca}, Radu V. Craiu\thanks{Department of Statistical Sciences, University of Toronto, Toronto, Ontario, Canada M5S 3G3. Email: craiu@utstat.utoronto.ca}
\ , and Jeffrey S. Rosenthal\thanks{Department of Statistical Sciences, University of Toronto, Toronto, Ontario, Canada M5S 3G3. Email: jeff@math.toronto.edu}}
\date{(March 2016; last revised March 2017)}

\begin{document}

\maketitle

\begin{abstract}

One of the most widely used samplers in practice is the component-wise Metropolis-Hastings (CMH) sampler that updates in turn the components of a vector valued Markov chain using accept-reject moves generated from a proposal distribution. When the target distribution of a Markov chain is irregularly shaped, a `good' proposal distribution for one part of the state space might be a `poor' one for another part of the state space. We consider a component-wise multiple-try Metropolis (CMTM) algorithm  that can automatically choose from  a set of candidate moves sampled from different distributions. The computational efficiency is increased using an adaptation rule for the CMTM algorithm that dynamically  builds a better set of proposal distributions as the Markov chain runs. The ergodicity of the adaptive chain is  demonstrated theoretically. The performance is studied via simulations and  data examples.

\end{abstract}


\section{Introduction}
\label{sec:intro}

Markov chain Monte Carlo (MCMC) methods are widely used to analyze complex probability distributions, especially within the Bayesian inference paradigm. One of the most used MCMC algorithms is the Metropolis-Hastings (MH) sampler, first developed by Metropolis et al. \citep{metropolis1953equation}, and later expanded by \cite{hastings1970monte}.  At each iteration the MH algorithm samples a  candidate new state from a proposal distribution which is subsequently accepted or rejected. When the state space of the chain is high dimensional or irregularly shaped, finding a good proposal distribution that can be used to update {\it all} the components of the chain simultaneously is very challenging, often impossible. The  optimality results for the acceptance rate of the Metropolis-Hastings algorithm \citep{gelman1996efficient, roberts2001optimal} have inspired the development of   the so-called {\it adaptive MCMC (AMCMC)} samplers that are designed to adapt their transition kernels based on the gradual  information about the target  that is collected through the very samples they produce. Successful designs can be found in   \cite{haario2001adaptive}, \cite{haario2006dram}, \cite{turro2007bgx}, \cite{roberts2009examples}, \cite{cry}, \cite{giordani2010adaptive}, and \cite{vihola2012robust} among others. Theoretical difficulties arise because the adaptive chains are no longer Markovian so ergodicity properties must be proven on a case-by-case basis. Attempts at streamlining the theoretical validation process for AMCMC samplers have been increasingly successful including
 \cite{atchade2005adaptive}, \cite{andrieu2006ergodicity}, \cite{andrieu2006efficiency}, \cite{roberts2007coupling}, \cite{fort2011convergence} and  \cite{craiu2014stability}.  
 For useful reviews of AMCMC we refer to \cite{tut-amcmc} and
\cite{roberts2009examples}.  
Despite many success stories, it is our experience that  existing adaptive strategies for MH  in high dimensional spaces may take a very long time to ``learn'' good simulation parameters so that the samplers may not improve much before the simulation is ended. 

We can  increase the computational efficiency if instead of using a full MH to update all the components at once, we choose  to update the components of the chain one-at-a-time. In this case the update rule follows the MH transition kernel but the acceptance or rejection is based on the target's conditional distribution of that component given all the other ones. More precisely, if we are interested in sampling from the continuous  density $\pi(x):  \mathcal{X} \subset \textbf{R}^d \rightarrow  \textbf{R}_{+} $; the component-wise MH (CMH)  transition kernel updates  the $i$th component of the chain, $x_{i}$, using a proposal $y_i \in \RR$, $y_i\sim T_{i}(\cdot|x_{i})$  and setting the next value of the chain as 
\[
z=\left \{ 
\begin{array}{cc}
(x_{1},\ldots,x_{i-1},y_{i},x_{i+1},\ldots, x_{d}) & \mbox{ w.p. } \alpha_{i}\\
x & \mbox{ w.p. }  1- \alpha_{i}\\
\end{array}
\right .
\]
where $$\alpha_{i}=\min\left\{ 1, { T(x_{i}|y_{i}) \pi(y_{i}|x_{[-i]}) \over T(y_{i}|x_{i})\pi(x_{i}|x_{[-i]})}\right\},$$
and $\pi(\cdot|x_{[-i]})$ is the target conditional distribution of the $i$th component given all the other components $x_{[-i]}=(x_{1},\ldots, x_{i-1},x_{i+1},\ldots, x_{d})$.  The CMH replaces the difficult problem of finding one good proposal in $d$ dimensions with the  easier problem of finding $d$ good 1-dimensional proposals. However, the latter task can also prove difficult if the conditional densities $\pi(\cdot|x_{[-i]})$ change significantly, e.g. have very different variances,  as 
$x_{[-i]}$ varies.  Intuitively, let us imagine that for a region of the sample space  of $x_{[-i]}$ the proposal $T_{i}$ must have a higher spread for the chain to mix well and a smaller one for the remaining part of the support. In this case an adaptive strategy based on a single proposal distribution cannot be efficient everywhere in the support of $\pi$.  Some success has been obtained  in lower dimensions or for distributions with a well-known structure using the regional adaptive MCMC strategies of  \cite{cry} or \cite{cbd}, but extending those approaches can be  cumbersome when $d$ is even moderately large. 
%and the conditional densities  of  $\pi$ change significantly, e.g. have very different variances, when we vary the conditioning variables. 
Other adaptive MCMC  ideas  proposed for the CMH too include
\cite{haario2005componentwise} where the authors propose to use component-wise random walk Metropolis (RWM) and to use the component-specific sample variance to tune the proposal's variance, along the same lines that were used by \cite{haario2001adaptive}  to adapt the proposal distribution for the joint RWM.  Another intuitive approach is  proposed in \cite{roberts2009examples} who aim for a particular acceptance rate for each component update. 

The strategy we propose here aims to close the gap that still exists between AMCMC and efficient CMH samplers.  When contemplating the problem, one may be tempted to try to ``learn'' each conditional distribution $\pi(\cdot| x_{[-i]})$, but parametric models are likely not flexible enough and nonparametric models will face the curse of dimensionality even for moderate values of $d$. Note that here the difficult part is understanding how the conditional distribution changes as $x_{[-i]}$ varies, which is a $(d-1)$-dimensional problem.

Before getting to the technical description of the algorithm, we
present here the intuitive idea behind our design.  Within the
CMH algorithm imagine that for each component we can propose
$m$ candidate moves, each generated from $m$ different proposal
distributions. Naturally, the latter will be selected to have
a diverse range of variances so that we generate some  proposals
close to the current location of the chain and some that are further
away. If we assume that the transition kernel for each component is
such that  among the proposed states it will select the one that is
most likely to lead to an acceptance, then one can reasonably infer
that this approach will improve the mixing of the chain provided that
the proposal distributions  are reasonably calibrated. To mirror the
discussion above, in a region where $T_{i}$ should have small spread,
one wants to have among the proposal distributions a majority with
small variances, and similarly in regions where $T_{i}$ should be
spread out we want to include among our proposal densities a majority
with larger variances.
This intuition can be tested using an approach
based on the multiple-try Metropolis (MTM) that originated
with \cite{liu2000multiple} and was further generalized by
\cite{casarin2013interacting}.

This paper is organized as follows.
Section \ref{sec:cmtm} introduces a component-wise multiple-try
Metropolis (CMTM) algorithm.  In Section \ref{sec:adpmtm}, we add
{\it adaption} to CMTM, creating a new ACMTM algorithm in which the
proposal distributions get modified on the fly according to the local
shape of the target distribution, and we prove (Theorem~\ref{convthm})
convergence of this algorithm.  Section \ref{sec: appl} then applies
the adaptive CMTM algorithm to numerical examples, and compares the
efficiency of the adaptive CMTM algorithm to other adaptive Metropolis
algorithms.

\section{Component-wise multiple-try Metropolis}
\label{sec:cmtm}

\subsection{Algorithm}
\label{sec:algo}

%The multiple-try Metropolis (MTM) was first introduced by Liu, Liang and Wong in \cite{liu2000multiple}. It was later expanded by Casarin, Craiu and Leison in \cite{casarin2013interacting} so a set of different proposal distributions can be used to generate multiple proposals at each iteration. 

Assume that a Markov chain $\{X_n\}$ is defined on \(\mathcal{X} \subset \textbf{R}^d\) with a target distribution $\pi$. The component-wise multiple-try Metropolis (CMTM) will update the chain one-component-at-a-time using $m$ proposals. Specifically,  the $k$th component of the chain is updated using proposals $\{y_j^{(k)}: \;\; 1\le j \le m\}$ that are sampled from 
$ \{T_j^{(k)} : \;\; 1\le j \le m\}$, respectively.  Let the value of the chain at iteration $n$ be \(X_n = x \in \RR^d\).  One  step of the CMTM involves updating every coordinate $X_k$ of the chain in a fixed order, for $k \in \{1,\ldots, d\}$. The following steps are performed to update the $k$th component:

\noindent
1. Draw proposals \(y_1^{(k)},\ldots,y_m^{(k)}\) where \(y_j ^{(k)}\sim T_j^{(k)}(\cdot | x_k)\). \\
2. Compute 
\begin{equation}
w_j^{(k)}(y_j^{(k)}, x) = \pi(y_j^{(k)}|x_{[-k]}) T_j(x_k |y_j^{(k)}) \lambda_j^{(k)} (y_j^{(k)}, x_k),
\label{weights}
\end{equation}
for each \(y_j^{(k)}\), where  $x_{[-k]}$ denotes the state of the chain without the $k$th component and \(\lambda_j^{(k)}(x_k,y_j^{(k)})\) is a nonnegative symmetric function satisfying \(\lambda_j^{(k)} (x_k, y_j^{(k)}) > 0\) whenever \(T_j^{(k)} (y_j^{(k)}|x_k) >0 \). \\
3. Select one \(y=y_s^{(k)}\) out of \(y_1^{(k)},\ldots,y_m^{(k)}\) with probabilities proportional to \(w_j(y_j^{(k)},x)\). \\
4. Draw \(x^{\ast (k)}_1,\ldots,x^{\ast (k)}_{s-1}, x^{\ast (k)}_{s+1},\ldots, x^{\ast (k)}_m\) where \(x_j^{\ast (k)}\sim T_j^{ (k)}(\cdot | y)\) and set \(x_s^{\ast (k)} = x\).\\
5. Accept $y$ with a probability
\begin{align*}
\rho = \min \Big[1, \frac{w_1(y_1^{(k)}, x)+\ldots+w_m(y_m^{(k)}, x)}{w_1(x^{\ast (k)}, y)+\ldots+w_m(x^{\ast (k)}_m, y)}\Big]
\end{align*}

\noindent 

We note that in step 1.\  the proposal distributions $T_j^{(k)}$  depend only on the $k$th component of the current state of the chain. More general formulations are possible, but make intuitive adaptive schemes more cumbersome and without clear benefits in terms of efficiency. Having dependent proposals can be beneficial  when the proposal distributions are identical \citep{craiu2007acceleration}.  However, in the current implementation the proposals have different scales so the advantage of using dependent proposals is less clear and will not be pursued in this paper.

Whether a proposal distribution is `good'  or not will depend on  the current state of the Markov chain, especially if the target distribution $\pi$ have conditional densities with varying properties, e.g.  different variances, across the target's support.  In addition to choosing the $m$ proposals, an added flexibility of the CMTM algorithm  is that we have freedom in choosing the  nonnegative  symmetric maps  $\lambda_j^{(k)}$ as long as they satisfy \(\lambda_j^{(k)} (x_k, y_j^{(k)}) > 0\) whenever \(T_j^{(k)} (y_j^{(k)}| x_k) >0 \). 
In subsequent sections we show that the CMTM algorithm can benefit from choosing a particular form of the function \(\lambda_j^{(k)} (x_k,y_j^{(k)})\).  

Our choice of $\lambda_j^{(k)}$ is guided by a simple and intuitive principle. 
Between two candidate moves $y_{1}^{(k)}$ and $y_{2}^{(k)}$ that are equally far from the current state we favour  \(y_1^{(k)}\) over \(y_2^{(k)}\) if  \(\pi(y_1^{(k)}|x_{[-k]})\) is greater than \(\pi(y_2^{(k)}|x_{[-k]})\), but  if \(\pi(y_1^{(k)}|x_{[-k]})\) is similar  to \(\pi(y_2^{(k)}|x_{[-k]})\), we would like CMTM to favour whatever candidate is further away from the current state. These simple rules lead us to consider
\begin{align}
\label{eq:lambdafn}
\lambda_j^{(k)}(x,y) = T_j^{(k)} (y_j^{(k)} | x_k)^{-1} \|(y_j^{(k)}-x_k)\|^\alpha, 
\end{align}
where \(\|\cdot\|\) is the Euclidean norm. Note that  this choice of $\lambda_j^{(k)}$ is possible because \(T_j^{(k)} (y_j^{(k)} | x_k)\) is a symmetric function in \(x_k\) and \(y_j^{(k)}\) as  it involves only one draw from a normal distribution with mean $x_k$.  

Replacing \eqref{eq:lambdafn} in  the weights equation \eqref{weights} results in
\begin{align}
w_j^{(k)}(y_j^{(k)}, x)&=\pi(y_j^{(k)}|x_{[-k]}) T_j^{(k)}(x_k |y_j^{(k)}) \lambda_j^{(k)} (y_j^{(k)}, x_k) \nonumber\\
& =\pi(y_j^{(k)}|x_{[-k]})\|(y_j^{(k)}-x_k)\|^\alpha.
\label{mod-weights}
\end{align}
With this choice of $\lambda$,   the selection probabilities are only dependent on the value of the target density  at the candidate point \(y_j^{(k)}\) and the size of the potential jump of the chain, were this candidate accepted. From \eqref{eq:lambdafn} we can see that the size of $\alpha$ will  balance of importance of the attempted jump distance from the current state over the importance of the candidate under $\pi$.  
However, while we understand the trade-off imposed by the choice of $\alpha$ for  selecting a candidate move, it is less clear how it will impact the overall performance of the CMTM, e.g acceptance rate or average jump distance. 

Therefore, it is paramount to gauge what are good choices for the parameter \(\alpha\) for the mixing of the CMTM chain. In the next section we approach this task via  the average squared jumping distance (ASJ) and the autocorrelation time  (ACT). To obtain the average squared jumping distance, we calculate the squared jumping distance for each iteration, \((X_{n+1}-X_n)^2\) and average them over the whole Markov chain run.  Note that if a new proposal is rejected and \((X_{n+1}-X_n)^2\)  is equal to zero, we still add zero to total sum of the squared jumping distances and divide the sum by the total number of iterations. The ACT can be calculated using 
\begin{align*}
\tau = 1+ 2 \sum_{k=1}^{\infty} \rho_k,
\end{align*}
where $\rho_k = \Cov(X_0, X_k)/Var(X_0)$ is the autocorrelation at lag $k$. Higher ACT for a Markov chain implies successive samples are highly correlated, which reduces the effective information contained in any given number of samples produced by the chain. 

While ACT has long been known to relate directly with the variance of the Monte Carlo estimators \citep{geyer-practicMCMC}, the ASJ incorporates both the jump distance and  the acceptance rate, a combination that has turned out to be useful in other AMCMC designs \citep[see for instance][]{cry}. Estimates of ACT and ASJ are  obtained by averaging over the realized path of the chain.  
\subsection{Choice of  \(\alpha\)}
\label{sec:optalpha}

In order to study the influence of the parameter $\alpha$ on the CMTM efficiency we  have conducted a number of simulation studies, some of which are described here.

We considered first a 2-dimensional mixture of two normal distributions

 \begin{minipage}{0.49\textwidth}
 \begin{align}
 0.5  N(\mu_1, \Sigma_1) + 0.5  N(\mu_2, \Sigma_2)
 \label{dist1}
 \end{align}
 where 
 \begin{align*}
 \begin{cases}
\mu_1 &= (5, 0)^T\nonumber\\
 \mu_2 &= (15, 0)^T\nonumber\\
 \Sigma_1 & = \mbox{diag} (6.25, 6.25)\nonumber\\
 \Sigma_2 & = \mbox{diag} (6.25, 0.25)\nonumber
 \end{cases}
 \end{align*}
 \end{minipage} \hspace{0.02\textwidth}
\begin{minipage}{0.49\textwidth}
\begin{figure}[H]
\vspace{-0.01\textheight}
\includegraphics[height=0.2\textheight, width=0.9\textwidth]{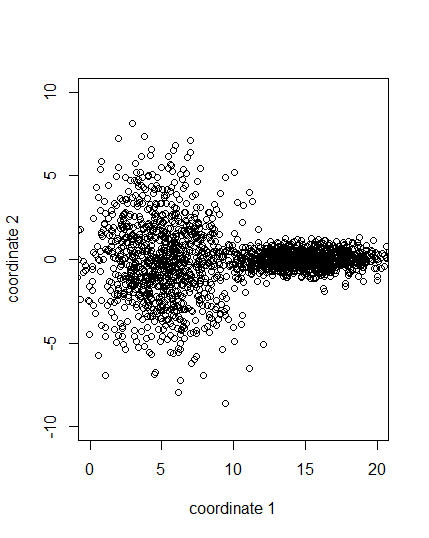}
\caption{Target density plot. 2-dimensional mixture of two normals}
\label{fig:2dimmix1}
\end{figure}
\end{minipage}
\vspace{0.001\textheight}

An iid sample of size $2000$  from \eqref{dist1} is  plotted in  Figure  \ref{fig:2dimmix1}. We run the CMTM algorithm repeatedly with \(\lambda_j ( x,y_{j})\) functions in \eqref{eq:lambdafn} while changing the value of \(\alpha\) from $0.1$ to $15$.  We choose $m=5$ as the number of proposals for each component,  while the proposal standard deviations  \(\sigma_{k,j}\)'s  are for each component  $1, 2, 4, 8$ and $16$. 
\vspace{0.001\textheight}\\
\begin{minipage}{0.49\textwidth}
\hspace{15pt} As we see in Figure \ref{fig:2dimmix1_propselect}, the proportion of each proposal distribution selected increases/decreases as \(\alpha\) changes. As expected, when $\alpha$ increases we see the  selection percentages of the proposal distributions with smaller \(\sigma_{k,j}\)'s drop   and those  with larger \(\sigma_{k,j}\)'s increase.  Figure \ref{fig:2dimmix1_propselect} shows, with larger \(\alpha\)'s, our algorithm favours proposal distributions with larger scales, which  makes sense based on the equation \eqref{mod-weights}.  
 \end{minipage} \hspace{0.02\textwidth}
\begin{minipage}{0.49\textwidth}
\begin{figure}[H]
\vspace{-0.01\textheight}
\includegraphics[height=0.25\textheight, width=0.9\textwidth]{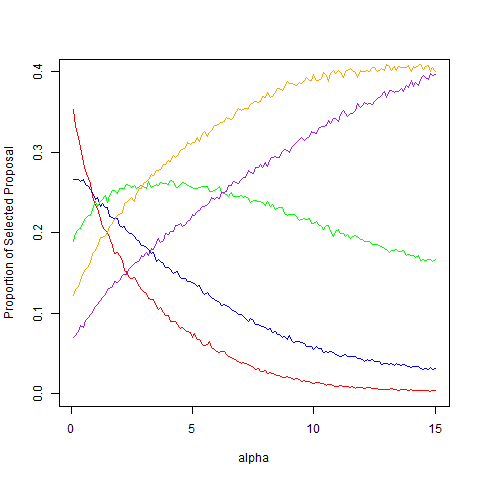}
\caption{Proportion of proposal distribution selected. Coordinate 1: Red, Blue, Green,  Orange and Purple lines show behaviour when
 \(\sigma_{k,j}= 1, 2, 4, 8 ,16\), respectively.}
\label{fig:2dimmix1_propselect}
\end{figure}
\end{minipage}
\vspace{0.001\textheight}

Figure \ref{fig:2dimmix1_alpha} shows how the ASJ and ACT   change as the value of \(\alpha\) changes.  We can infer that the highest efficiency is achieved for $\alpha \in (2,4)$.     

\begin{figure}[H]
\includegraphics[width=0.95\textwidth]{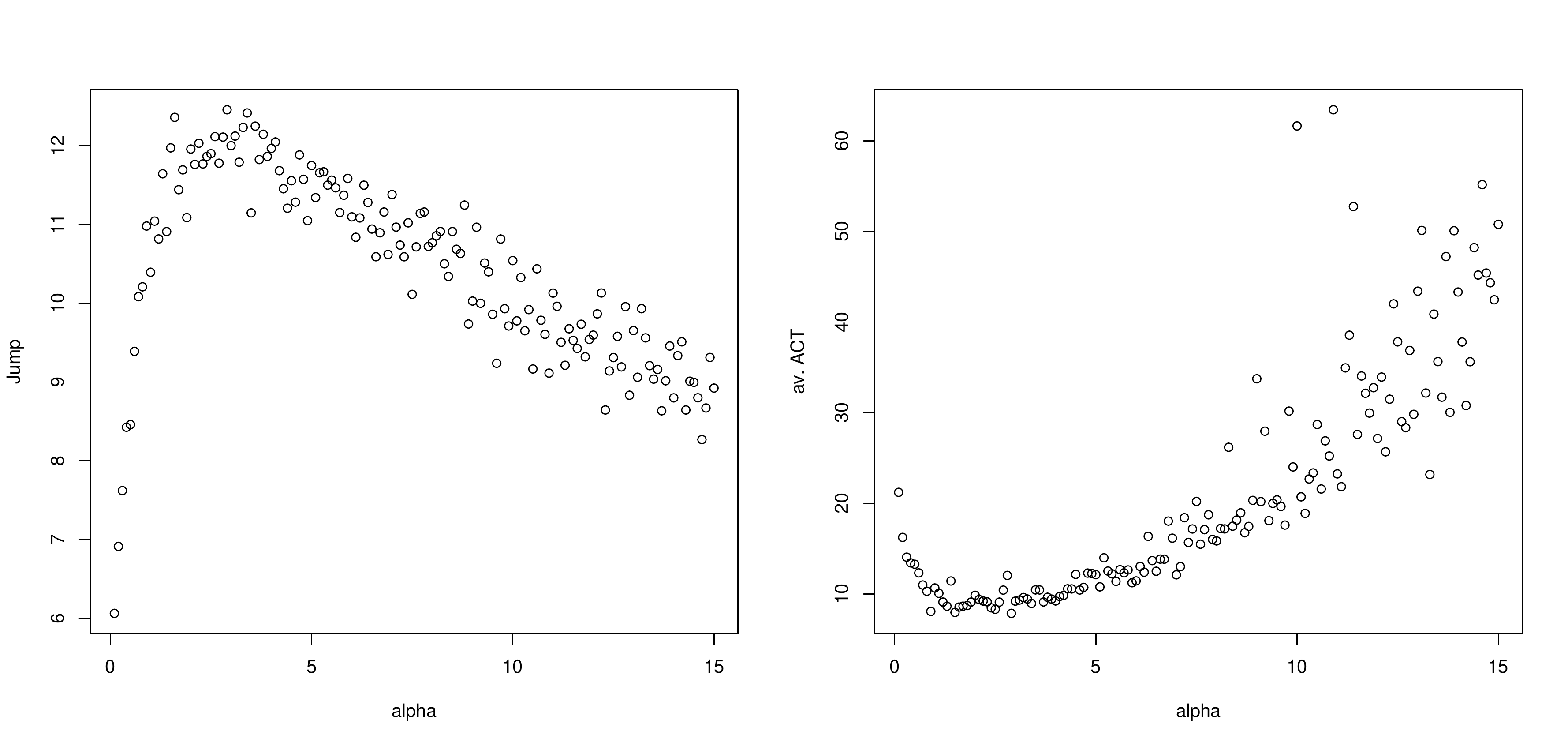}
\caption{Two-Dimensional Mixture of  two Gaussians: ASJ  (left panel) and ACT (right panel) for different values of  \(\alpha\).  For each $\alpha$, the estimates are obtained from a single run with 100,000 iterations.}
\label{fig:2dimmix1_alpha}
\end{figure}

We also examined a  4-dimensional mixture of two normal distributions as our target density:
 \begin{align*}
0.5  N(\mu_1, \Sigma_1) + 0.5 N(\mu_2, \Sigma_2),
 \end{align*}
 where 
 \begin{align*}
 \begin{cases}
 \mu_1 &= (5, 5, 0, 0)^T\nonumber\\
 \mu_2 &= (15, 15, 0, 0)^T\nonumber\\
 \Sigma_1 & = \mbox{diag} (6.25, 6.25, 6.25, 0.01)\nonumber\\
 \Sigma_2 & =  \mbox{diag} (6.25, 6.25, 0.25, 0.01).\nonumber
 \end{cases}
 \end{align*}

 The number of proposals, $m=5$ and \(\sigma_{k,j}\)'s of the set of proposal distributions for each coordinate are $0.5, 1, 2, 4$ and $8$. Figure \ref{fig:4dimmix1_alpha} shows the results. We notice that the ACT measurements are more noisy, while the ASJ ones yield  a more precise message that is in line with the previous example. Once again we can see from  Figure \ref{fig:4dimmix1_alpha} that the average squared jumping distances are largest for  $\alpha\in(2,4)$. 

\begin{figure}[H]
\includegraphics[width=0.95\textwidth]{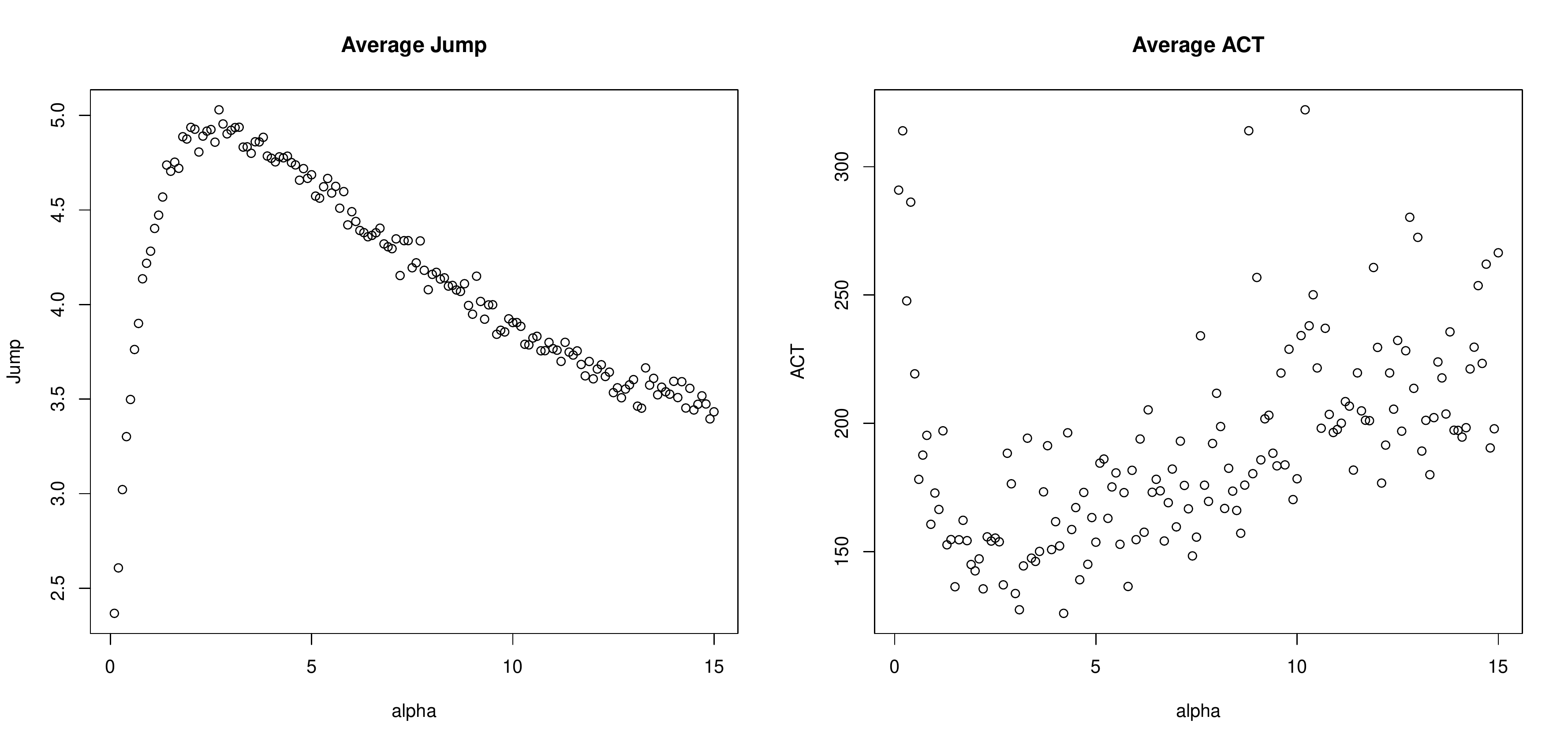}
\caption{4-Dimensional Mixture of two Gaussians: ASJ  (left panel) and ACT (right panel) for different values of  \(\alpha\).  For each $\alpha$, the estimates are obtained from a single run with 100,000 iterations.}
\label{fig:4dimmix1_alpha}
\end{figure}

Other numerical experiments not reported here agree with the two examples presented and suggest that optimal values of $\alpha$ are between $2$ and $4$.  In the absence of theoretical results we cannot claim a universal constant $\alpha$ that would be optimal  in every example. However, based on the available evidence, we believe that a value of $\alpha$ in the $(2,4)$ range will increase the efficiency of the chain.  Henceforth we fix  
$\alpha=2.9$ in all simulations involving CMTM. 

\section{Adaptive Component-wise multiple-try Metropolis}
\label{sec:adpmtm}

\subsection{CMTM Favours Component-wise `Better' Proposal Distributions}
\label{sec:selectprob}

The intuition behind our construction as described in the Introduction, relies on the idea that CMTM will automatically tend to choose the ``right'' proposal among the $m$ possible ones. In this section we verify empirically that this is indeed the case.

 We consider the same  4-dimensional mixture of normal distributions from Section \ref{sec:optalpha} as our target distribution and run the CMTM algorithm. The target parameters are set to reflect the numerical experiments reported in Section \ref{sec: appl}, i.e.  $m=20$ and   $\sigma_{k,j}=2^j$ with $j \in \{-10,-9,\ldots,9\}$. Table \ref{tab: mtm_4dimmixnorm_pselect} reports the selection probabilities computed from  10,000 samples for each proposal and each coordinate. 
 
\begin{table}[h]
\centering
\begin{tabular}{rrrrr}
  \hline
  &\multicolumn{4}{c}{Coordinate}\\
\hhline{~----}
\(\sigma_{k,j}\) & coord1 & coord2 & coord3 & coord4 \\ 
  \hline
$2^{-10}$ & 0.00 & 0.00 & 0.00 & 0.00  \\ 
$2^{-9}$  & 0.00 & 0.00 & 0.00 & 0.00  \\ 
$2^{-8}$  & 0.00 & 0.00 & 0.00 & 0.00  \\ 
$2^{-7}$  & 0.00 & 0.00 & 0.00 & 0.00  \\ 
$2^{-6}$  & 0.00 & 0.00 & 0.00 & 0.00  \\ 
$2^{-5}$  & 0.00 & 0.00 & 0.00 & 0.03  \\ 
$2^{-4}$  & 0.00 & 0.00 & 0.00 & 0.11  \\ 
$2^{-3}$  & 0.00 & 0.00 & 0.01 & 0.25  \\ 
$2^{-2}$  & 0.00 & 0.00 & 0.03 & 0.27  \\ 
$2^{-1}$  & 0.01 & 0.01 & 0.11 & 0.17  \\ 
$2^{0}$   & 0.05 & 0.05 & 0.15 & 0.08  \\ 
$2^{1}$   & 0.15 & 0.14 & 0.19 & 0.04  \\ 
$2^{2}$   & 0.26 & 0.26 & 0.20 & 0.02  \\ 
$2^{3}$   & 0.24 & 0.25 & 0.15 & 0.01  \\ 
$2^{4}$   & 0.14 & 0.14 & 0.08 & 0.01  \\ 
$2^{5}$   & 0.08 & 0.07 & 0.04 & 0.00  \\ 
$2^{6}$   & 0.04 & 0.04 & 0.02 & 0.00  \\ 
$2^{7}$   & 0.02 & 0.02 & 0.01 & 0.00  \\ 
$2^{8}$   & 0.01 & 0.01 & 0.01 & 0.00  \\ 
$2^{9}$   & 0.00 & 0.00 & 0.00 & 0.00  \\ 
   \hline
\end{tabular}
\caption{CMTM: Frequency of selection for each proposal and each coordinate.}
\label{tab: mtm_4dimmixnorm_pselect}
\end{table}

\begin{table}[H]
\setlength{\tabcolsep}{5pt}
    \begin{subtable}{.48\textwidth}
\caption{$X_{n,2} < 8$}
\centering
\begin{tabular}{r|rrrr}
  \hline
  &\multicolumn{4}{c}{Coordinate}\\
\hhline{~----}
\(\sigma_{k,j}\) & coord1 & coord2 & coord3 & coord4 \\ 
  \hline
$2^{-10}$ 	&	 0.00	&	 0.00	&	 0.00	&	 0.00	\\
$2^{-9}$ 	&	 0.00	&	 0.00	&	 0.00	&	 0.00	\\
$2^{-8}$	&	 0.00	&	 0.00	&	 0.00	&	 0.00	\\
$2^{-7}$ 	&	 0.00	&	 0.00	&	 0.00	&	 0.00	\\
$2^{-6}$  	&	 0.00	&	 0.00	&	 0.00	&	 0.01	\\
$2^{-5}$ 	&	 0.00	&	 0.00	&	 0.00	&	 0.03	\\
$2^{-4}$ 	&	 0.00	&	 0.00	&	 0.00	&	 0.10	\\
$2^{-3}$  	&	 0.00	&	 0.00	&	 0.00	&	 0.25	\\
$2^{-2}$  	&	 0.00	&	 0.00	&	 {\bf 0.00}	&	 0.27	\\
$2^{-1}$ 	&	 0.01	&	 0.01	&	 {\bf 0.01}	&	 0.18	\\
$2^{0}$ 	&	 0.05	&	 0.05	&	 {\bf 0.04}	&	 0.09	\\
$2^{1}$ 	&	 0.16	&	 0.14	&	 {\bf 0.17}	&	 0.04	\\
$2^{2}$ 	&	 0.27	&	 0.26	&	 {\bf 0.28}	&	 0.02	\\
$2^{3}$ 	&	 0.24	&	 0.25	&	 {\bf 0.23}	&	 0.01	\\
$2^{4}$ 	&	 0.13	&	 0.14	&	 {\bf 0.13}	&	 0.01	\\
$2^{5}$ 	&	 0.07	&	 0.08	&	 0.07	&	 0.00	\\
$2^{6}$ 	&	 0.03	&	 0.04	&	 0.03	&	 0.00	\\
$2^{7}$  	&	 0.02	&	 0.02	&	 0.01	&	 0.00	\\
$2^{8}$ 	&	 0.01	&	 0.00	&	 0.01	&	 0.00	\\
$2^{9}$ 	&	 0.00	&	 0.00	&	 0.01	&	 0.00	\\
   \hline
\end{tabular}
    \end{subtable}%
    \hspace{0.04\textwidth}
    \begin{subtable}{.48\textwidth}
\caption{$X_{n,2}\geq 8$}
\label{tab:2dimmix1_pselect2}
\centering
\begin{tabular}{r|rrrr}
  \hline
  &\multicolumn{4}{c}{Coordinate}\\
\hhline{~----}
\(\sigma_{k,j}\) & coord1 & coord2 & coord3 & coord4 \\ 
  \hline
$2^{-10}$ 	&	 0.00	&	 0.00	&	 0.00	&	 0.00	\\
$2^{-9}$ 	&	 0.00	&	 0.00	&	 0.00	&	 0.00	\\
$2^{-8}$	&	 0.00	&	 0.00	&	 0.00	&	 0.00	\\
$2^{-7}$ 	&	 0.00	&	 0.00	&	 0.00	&	 0.00	\\
$2^{-6}$  	&	 0.00	&	 0.00	&	 0.00	&	 0.00	\\
$2^{-5}$ 	&	 0.00	&	 0.00	&	 0.00	&	 0.04	\\
$2^{-4}$ 	&	 0.00	&	 0.00	&	 0.00	&	 0.12	\\
$2^{-3}$  	&	 0.00	&	 0.00	&	 0.02	&	 0.24	\\
$2^{-2}$  	&	 0.00	&	 0.00	&	 {\bf 0.06}	&	 0.26	\\
$2^{-1}$ 	&	 0.01	&	 0.01	&	 {\bf 0.20}	&	 0.17	\\
$2^{0}$ 	&	 0.05	&	 0.06	&	 {\bf 0.24}	&	 0.08	\\
$2^{1}$ 	&	 0.14	&	 0.14	&	 {\bf 0.20}	&	 0.04	\\
$2^{2}$ 	&	 0.26	&	 0.26	&	 {\bf 0.13}	&	 0.02	\\
$2^{3}$ 	&	 0.24	&	 0.25	&	 {\bf 0.08}	&	 0.01	\\
$2^{4}$ 	&	 0.14	&	 0.14	&	 {\bf 0.03}	&	 0.00	\\
$2^{5}$ 	&	 0.09	&	 0.07	&	 0.02	&	 0.00	\\
$2^{6}$ 	&	 0.04	&	 0.04	&	 0.01	&	 0.00	\\
$2^{7}$  	&	 0.02	&	 0.02	&	 0.01	&	 0.00	\\
$2^{8}$ 	&	 0.01	&	 0.01	&	 0.00	&	 0.00	\\
$2^{9}$ 	&	 0.00	&	 0.00	&	 0.00	&	 0.00	\\
   \hline
\end{tabular}
    \end{subtable} 
        \caption{Selection frequencies for each proposal  and each coordinate calculated  on two  regions of the support, $A_1= \{X \in \RR^4: \; X_2 <8\}$ and $A_2=\{X \in \RR^4: \; X_2  \ge 8\}$. The entries in boldface show the difference in selection frequencies for some of the proposals in  the two regions of the support considered.}
    \label{tab:2dimmix1_pselect}
\end{table}

Tables \ref{tab:2dimmix1_pselect}  and \ref{tab:2dimmix1_pselect2} present the proportion of candidate selection and acceptance rates for each proposal.  We compare the proportion of proposals selected in  the regions  \(A_1=\{X_{n,2} < 8\} \) and \(A_2=\{X_{n,2} \geq 8\} \). While these regions are defined based on knowing the target exactly, they do not enter in any way in the design of the CMTM and are used here only to verify that the sampler indeed automatically adapts to local characteristics of the target.  We can see that the CMTM favours  proposal distributions with smaller \(\sigma_{k,j}\)'s when updating the third coordinate in the region $A_2$. This  is appropriate given that in that region larger moves for the third coordinate will tend to be rejected.  This pattern does not hold for the first two coordinates  for which larger moves are appropriate throughout the sample space. This is in line with what is expected since the target variances ($=6.25$) are the same in both directions in that region and confirms  that  the CMTM algorithm tends to choose the `better' proposal distribution out of the available choices provided at each iteration. 

\subsection{Comparison with a Mixture Transition Kernel}

An astute reader may wonder about a different strategy for using the  different proposals that one may have at one's disposal. Maybe  the most natural alternative is a random mixture of the component-wise Metropolis-Hastings (CMH) algorithms. The set of proposal distributions used in both algorithms is the same and we assign equal weights for the proposal distributions in the mixture.  The mixture CMH  kernel  selects each proposal at random with equal probability, but since a single proposal is produced each time a coordinate is updated, it is different than a CMTM algorithm with equal weights $w_j$.

However, this comparison will help us determine whether  adjusting the selection probabilities of each proposal distribution is an improvement over  equal probability selection. Our target distribution is the 4-dimensional mixture of two normals introduced in Section \ref {sec:optalpha}. We use $m=20$ and the same proposal scales discussed in the previous section. In Tables \ref{tab: 4dimmixnorm_acrate} and \ref{tab: mtm_4dimmixnorm_acrate} we present the acceptance rates for each coordinate and each proposal for the two samplers.   The results in Table \ref{tab: 4dimmixnorm_acrate} suggest that proposal distributions with small variances have their proposals, if selected, accepted with with high frequency. In the case of mixture of CMH this also means that  if we were to guide our selection of proposals based on  acceptance rates, we would  favour  small jumps.  The selection step in the CMTM seems to balance out a lot more the acceptance frequencies for the proposals used.  The even acceptance frequencies mean that they are not very informative about which variances are to be used in each coordinate.

To compare the efficiency of  the two algorithms, we report in Table \ref{tab: 4dimmixnorm_perf}  the ASJ and ACT calculated from 100 replicated runs as well as  the CPU time. We note that the average squared jumping distance significantly improves with the CMTM compared to the  mixture CMH. We can also see that  for all the chain's coordinates the ACT is an order of magnitude smaller for the CMTM than  the  mixture CMH.   When programming the examples in this paper we were able to take advantage of  the  R software's efficient handling of vector operations. This explain the small difference in CPU time even as CMTM requires   $m$ times more evaluations of the target than  the mixture CMH. 

\begin{table}[H]
\setlength{\tabcolsep}{5pt}
      \begin{subtable}{.48\textwidth}
    \caption{Mixture of CMH}
\label{tab: mwg_4dimmixnorm_acrate}
\centering
\begin{tabular}{rrrrr}
  \hline
  &\multicolumn{4}{c}{Coordinate}\\
\hhline{~----}
\(\sigma_{k,j}\) & coord1 & coord2 & coord3 & coord4 \\ 
  \hline
$2^{-10}$ 	&	 1.00	&	 1.00	&	 1.00	&	 1.00	\\
$2^{-9}$ 	&	 1.00	&	 1.00	&	 1.00	&	 0.99	\\
$2^{-8}$	&	 1.00	&	 1.00	&	 1.00	&	 0.98	\\
$2^{-7}$ 	&	 0.99	&	 1.00	&	 1.00	&	 0.98	\\
$2^{-6}$  	&	 1.00	&	 1.00	&	 0.99	&	 0.93	\\
$2^{-5}$ 	&	 0.99	&	 1.00	&	 1.00	&	 0.90	\\
$2^{-4}$ 	&	 0.99	&	 0.99	&	 0.98	&	 0.78	\\
$2^{-3}$  	&	 0.99	&	 0.97	&	 0.96	&	 0.65	\\
$2^{-2}$  	&	 0.97	&	 0.95	&	 0.97	&	 0.39	\\
$2^{-1}$ 	&	 0.91	&	 0.94	&	 0.88	&	 0.23	\\
$2^{0}$ 	&	 0.88	&	 0.87	&	 0.77	&	 0.11	\\
$2^{1}$ 	&	 0.76	&	 0.76	&	 0.63	&	 0.06	\\
$2^{2}$ 	&	 0.58	&	 0.58	&	 0.43	&	 0.04	\\
$2^{3}$ 	&	 0.39	&	 0.36	&	 0.26	&	 0.01	\\
$2^{4}$ 	&	 0.21	&	 0.21	&	 0.19	&	 0.01	\\
$2^{5}$ 	&	 0.11	&	 0.12	&	 0.11	&	 0.00	\\
$2^{6}$ 	&	 0.05	&	 0.05	&	 0.04	&	 0.00	\\
$2^{7}$  	&	 0.02	&	 0.04	&	 0.02	&	 0.00	\\
$2^{8}$ 	&	 0.02	&	 0.00	&	 0.01	&	 0.00	\\
$2^{9}$ 	&	 0.01	&	 0.01	&	 0.00	&	 0.00	\\
   \hline
\end{tabular}
    \end{subtable}%
    \hspace{0.04\textwidth}
    \begin{subtable}{.48\textwidth}
\caption{CMTM}
\label{tab: mtm_4dimmixnorm_acrate}
\centering
\begin{tabular}{rrrrr}
  \hline
  &\multicolumn{4}{c}{Coordinate}\\
\hhline{~----}
\(\sigma_{k,j}\) & coord1 & coord2 & coord3 & coord4 \\ 
  \hline
$2^{-10}$ 	&	  NaN	&	  NaN	&	  NaN	&	  NaN	\\
$2^{-9}$ 	&	  NaN	&	  NaN	&	  NaN	&	  NaN	\\
$2^{-8}$	&	  NaN	&	  NaN	&	  NaN	&	  NaN	\\
$2^{-7}$ 	&	  NaN	&	  NaN	&	  NaN	&	 0.17	\\
$2^{-6}$  	&	  NaN	&	  NaN	&	  NaN	&	 0.52	\\
$2^{-5}$ 	&	  NaN	&	  NaN	&	 1.00	&	 0.44	\\
$2^{-4}$ 	&	 0.50	&	  NaN	&	 0.50	&	 0.52	\\
$2^{-3}$  	&	 0.00	&	 0.00	&	 0.42	&	 0.50	\\
$2^{-2}$  	&	 0.17	&	 0.43	&	 0.53	&	 0.47	\\
$2^{-1}$ 	&	 0.49	&	 0.38	&	 0.58	&	 0.47	\\
$2^{0}$ 	&	 0.54	&	 0.45	&	 0.49	&	 0.44	\\
$2^{1}$ 	&	 0.57	&	 0.52	&	 0.52	&	 0.45	\\
$2^{2}$ 	&	 0.51	&	 0.49	&	 0.49	&	 0.37	\\
$2^{3}$ 	&	 0.48	&	 0.45	&	 0.47	&	 0.41	\\
$2^{4}$ 	&	 0.46	&	 0.45	&	 0.48	&	 0.33	\\
$2^{5}$ 	&	 0.41	&	 0.48	&	 0.48	&	 0.33	\\
$2^{6}$ 	&	 0.40	&	 0.35	&	 0.50	&	 0.43	\\
$2^{7}$  	&	 0.45	&	 0.31	&	 0.45	&	 0.38	\\
$2^{8}$ 	&	 0.47	&	 0.24	&	 0.35	&	 0.00	\\
$2^{9}$ 	&	 0.33	&	 0.45	&	 0.61	&	  NaN	\\
   \hline
\end{tabular}
    \end{subtable} 
      \caption{Post-selection acceptance frequencies. The NA's in the table are due to the fact that some proposals are never selected  for some of the coordinates . }
    \label{tab: 4dimmixnorm_acrate}
\end{table}

\begin{table}[H]
\setlength{\tabcolsep}{5pt}
     \begin{subtable}{.48\textwidth}
    \caption{Mixture of CMH}
\label{tab: mwg_4dimmixnorm_perf}
\centering
\begin{tabular}{rrrrrr}
  \hline
& Min. & Median & Mean & Max.  \\ 
  \hline
  cputime(s)& 4.47 & 4.56 & 4.57 & 4.97 \\
sq. jump& 0.467 & 0.619  &0.622 & 0.784   \\
\hline
\hline
&coord1 & coord2 & coord3 & coord4\\
\hline
ACT& 464.21 & 460.41 & 28.07 & 26.70 \\ 
   \hline
   \end{tabular}
    \end{subtable}%
    \hspace{0.04\textwidth}
    \begin{subtable}{.48\textwidth}
\caption{CMTM}
\label{tab: mtm_4dimmixnorm_perf}
\centering
\begin{tabular}{rrrrrr}
  \hline
 & Min. & Median & Mean & Max.  \\ 
  \hline
 cputime(s) & 10.25 & 10.41 & 10.43 & 11.22 \\
sq. jump & 6.20 & 6.62 &  6.62 & 7.07 \\ 
\hline
\hline
& coord1 & coord2 & coord3 & coord4\\
\hline
ACT& 41.96 & 41.25 & 1.64 & 1.64\\ 
   \hline
        \end{tabular}
    \end{subtable} 
     \caption{Comparison of performance indicators that were computed from 100 independently replicated runs. The tables contain statistics about the execution time for a complete run (cputime), the average squared jump distance and the ACT.}
    \label{tab: 4dimmixnorm_perf}
\end{table}

\subsection{The Adaptive CMTM Algorithm (ACMTM)}
\label{sec:algomtm}

Given its propensity to choose the best candidate put forward by the proposal distributions, it is reasonable to infer that CMTM's performance will be roughly aligned with the most suitable proposal for the  region where the chain current state lies. The other side of the coin is that  a whole set of bad proposals will compromise the efficiency of the CMTM algorithm. Therefore, we focus our efforts in developing an adaptive CMTM (AMCTM) design that aims to minimize, possibly annihilate, the chance of having at our disposal only poorly calibrated proposal distributions in any region of the space.

The adaptation  strategy is centered on finding well-calibrated values for the set $S_{k}=\{\sigma_{k,j}: \; 1\le j \le m\}$ for every coordinate $1\le k \le d$.  Note that $S_{k}$ varies across coordinates.

Consider an arbitrarily fixed coordinate $k$ and suppose we label the
$m$ proposal distributions such that $\sigma_{k,1}<\sigma_{k,2}<\ldots
<\sigma_{k,m}$.  Changes in the kernel occur at fixed points in
the simulation process, called {\it adaption points}.  We want
our adaptive algorithm to adapt less and less as the simulation
proceeds, a condition known as {\it Diminishing Adaptation (DA)}
and long recognized as being useful for establishing the chain's
valid asymptotic behaviour \citep{roberts2007coupling}. However, the
adaption strategy proposed above may not diminish in the long run,
so we ensure the DA condition more directly by only adapting on $a$th
iteration (for $a \ge 1$) with probability $P_a=\max (0.99^{a-1},
\frac{1}{\sqrt{a}})$.  Since $P_a \to 0$, the DA condition is ensured.
On the other hand, we chose $P_a$ so that it is decreases slowly and
has high values at the beginning of the run when most adaptations will
take place. Furthermore, the Borel-Cantelli lemma guarantees that the
adaption will keep occurring for as long as we run the chain since
$\sum_{a=1}^{\infty} P_a =\infty$.

An adaption is required for the standard deviations $\sigma_{k,j}$  only if we notice that the candidates generated by the proposal distribution $T_j^{(k)}$ with the smallest scale, $\sigma_{k,1}$, or the largest one, $\sigma_{k,m}$, are under- or over-selected.   For instance, suppose  that in an inter-adaptation time interval  the candidates generated using $\sigma_{k,1}$ are selected more than $100\times{2 \over m}$\% or less that $100\times{1 \over 2m}$\% of the time.  If we denote $q_{j}$ the frequency of selecting the candidate generated using $\sigma_{k,j}$ we have $m \max q_j \ge \sum_j q_{j}=1\ge m \min q_{j}$. Thus, the  thresholds represent, respectively, more than double  the  selection percentage for the least selected proposal and  less than half of the selection percentage for the most popular proposal.  A high selection percentage for $\sigma_{k,1}$ suggests that the chain tends to favour, when updating the $k$th coordinate, proposals with smaller scale so the ACMTM design requires to: 1)  halve the value of $\sigma_{k,1}$;  2) recalculate the intermediate values, $\sigma_{k,2},\ldots, \sigma_{k,m-1}$  to be equidistant between $\sigma_{k,1}$ and $\sigma_{k,m}$ on the log-scale.  A low selection percentage for $\sigma_{k,1}$ will ensure that the lowest scale is doubled up followed by step 2).  

Similarly, if the largest 
element in $S_{k}$, $\sigma_{k,m}$,  produces proposals with selection percentages above or below the thresholds mentioned above, we will   double or halve  $\sigma_{k,m}$, respectively. Each modification is followed by redistribution of the intermediate scales.

 If neither the smallest nor the largest elements in $S_{k}$ produce proposals that are   outside the boundaries set by the two thresholds, we wait until the algorithm reaches the next `adaption point' and recalculate  the proportion of each proposal  candidate being selected during the last inter-adaption time interval. The pseudo-code for the ACMTM is presented in Algorithm \ref{pc-acmtm}.

\begin{algorithm}
   \caption{Adaptive CMTM}
    \begin{algorithmic}
     \State Given:
            \begin{itemize}
             \item $M$ - number of MCMC iterations
             \item $m$ - number of proposals
             \item $d$ - number of coordinates
             \item $\{\sigma_{k,j}:1\leq k \leq d, 1\leq j \leq m \}$ - initial proposals
             \end{itemize}
    \State Set initial values:
           \begin{itemize}
            \item $\beta=100$ - the number of iterations between attempting an adaptation
            \item $P_a=1$ - probability of adapting at each attempt
            \end{itemize}
    \For{$t = 1$ to $M$}
       \If{ $t = 0$ $\mod \beta$ }
       \State $a=t/\beta$
         \State $u \sim \ru[0,1]$
          \If{ $u\leq P_a$ }
            \State Let $\sigma_{k,j} \le \ldots \le \sigma_{k,m}$ be the scales used and $\{S_{k,j}: \; 1\le k \le d, 1\le j \le m \}$ be the selection rates computed since the previous adaptation till now. Then
            \For{ $k=1$ to $d$ }
              \If{ $S_{k,m} > 2/m$ }
                 \State $\sigma_{k,m}=2 \sigma_{k,m}$
                 \State Make $\{\sigma_{k,j}\}$ equidistant on log base 2 scale
              \ElsIf { $(S_{k,m} < 1/(2m)) \land (\sigma_{k,1}<\sigma_{k,m}/2)}$
                 \State $\sigma_{k,m}=\sigma_{k,m}/2$
                 \State Make $\{\sigma_{k,j}\}$ equidistant on log base 2 scale
              \EndIf

              \If{ $S_{k,1} > 2/m$ }
                 \State $\sigma_{k,1}= \sigma_{k,1}/2$
                 \State Make $\{\sigma_{k,j}\}$ equidistant on log base 2 scale
              \ElsIf { $(S_{k,1} < 1/(2m)) \land (2 \sigma_{k,1}<\sigma_{k,m})}$
                 \State $\sigma_{k,1}=2 \sigma_{k,1}$
                 \State Make $\{\sigma_{k,j}\}$ equidistant on log base 2 scale
              \EndIf
            \EndFor
          \EndIf
        \State $P_a=\max (0.99^{a-1}, \frac{1}{\sqrt{a}})$
       \EndIf
      \State Perform CMTM move
    \EndFor
\end{algorithmic}
\label{pc-acmtm}
\end{algorithm}  

Finally, we make two minor technical modifications to our
ACMTM algorithm, to ensure the Containment condition of
\cite{roberts2007coupling}, and thus allow us to prove the convergence
of our algorithm in Section \ref{sec: convergence} below.  Namely:

\medskip \noindent \bf (A1) \rm
We choose a (very large) non-empty compact subset \(K \subset \mathcal{X}
\), and force $X_n \in K$ for all $n$.  Specifically, we
reject all proposals \(Y_{n+1} \not \in K\) (but if \(Y_{n+1}
\in K\), then we still accept/reject \(Y_{n+1}\) by the usual rule
for the CMTM algorithm described in Section \ref{sec:algo}).
Correspondingly, the initial value $X_0$ should be chosen in $K$.

\medskip \noindent \bf (A2) \rm
We choose a (very large) constant \(L > 0\) and a (very small) constant
\(\epsilon>0\), and force the proposal scalings
\(\sigma_{k,j}\) to always be in $[\epsilon,L]$.  Specifically,
if \(\sigma_{n,k,j}\) is the value of \(\sigma_{k,j}\) used at the
$n$-th iteration in our adaptive CMTM algorithm, then if $\sigma_{n,k,j}$
would be greater than L, we instead set $\sigma_{n,k,j}= L$, while
if $\sigma_{n,k,j}$ would be less than $\epsilon$, we instead set
$\sigma_{n,k,j}= \epsilon$. Correspondingly, the initial values
$\sigma_{0,k,j}$ should all be chosen in $[\epsilon,L]$.

\medskip
\begin{remark}
Our adaptive algorithm keeps the number of different
proposals at each iteration fixed at some constant $m$.  We have also
experimented with allowing the value $m$ itself to be updated adaptively.
This works fairly well, but does not appear to
offer any clear improvement over keeping $m$ constant, so we do not
pursue it further herein.
However, our theoretical justification also covers this case as long
as the possible $m$ values are bounded; see the remark following the
proof of Theorem~\ref{convthm} below.
\end{remark}

\subsection{To Adapt or Not To Adapt?}

We compare the ACMTM algorithm with the CMTM algorithm without adaption to see if the adaption indeed improves the efficiency of the algorithm. We use the 4-dimensional mixture of two normal distributions from Section \ref{sec:optalpha} as our target distribution. The \(\sigma_{k,j}\)'s for the non-adaptive algorithm are those given in Section 3.1  and they are also the starting $\sigma_{k,j}$'s for the adaptive algorithm. Evidently the final values are the same as the initial ones for the non-adaptive version of the sampler. In Table \ref{tab: acmtm_4dimmixnorm_sigmaj} we report  the final  values of the $\sigma_{k,j}$'s  obtained after the last adaption in one random  run of ACMTM.   For this particular run, the last adaption occurred right after 1800 iterations out of 10000 iterations in total. We notice that the scales chosen vary from component to component. For instance, the fourth component of the chain has a smaller marginal variance so the adaption  will favour smaller scales.  Similarly, the third  component requires both large and small proposal scales and we can see that reflected in the range of values for  $\{\sigma_{3,j}; \; 1\le j \le m\}$ which is different than for the first two components.

The comparison in terms of ASJ and ACT  is based on 100 independent replicates. The results shown in Table \ref{tab: 4dimmixnorm_perf_acmtm} indeed confirm the benefits of adaptation, as both ASJ and ACT are in agreement regarding the superiority of ACMTM over CMTM.

    \begin{table}
\centering
\begin{tabular}{r|rrrr}
  \hline
& coord1 & coord2 & coord3 & coord4 \\ 
  \hline
prop1	&	 4.0000	&	 4.0000	&	 2.0000	&	 0.1250	\\
prop2	&	 4.1486	&	 4.1486	&	 2.0743	&	 0.1345	\\
prop3	&	 4.3028	&	 4.3028	&	 2.1514	&	 0.1446	\\
prop4	&	 4.4626	&	 4.4626	&	 2.2313	&	 0.1556	\\
prop5	&	 4.6284	&	 4.6284	&	 2.3142	&	 0.1674	\\
prop6	&	 4.8004	&	 4.8004	&	 2.4002	&	 0.1800	\\
prop7	&	 4.9788	&	 4.9788	&	 2.4894	&	 0.1937	\\
prop8	&	 5.1638	&	 5.1638	&	 2.5819	&	 0.2083	\\
prop9	&	 5.3556	&	 5.3556	&	 2.6778	&	 0.2241	\\
prop10	&	 5.5546	&	 5.5546	&	 2.7773	&	 0.2410	\\
prop11	&	 5.7610	&	 5.7610	&	 2.8805	&	 0.2593	\\
prop12	&	 5.9750	&	 5.9750	&	 2.9875	&	 0.2789	\\
prop13	&	 6.1970	&	 6.1970	&	 3.0985	&	 0.3000	\\
prop14	&	 6.4273	&	 6.4273	&	 3.2136	&	 0.3227	\\
prop15	&	 6.6661	&	 6.6661	&	 3.3330	&	 0.3472	\\
prop16	&	 6.9138	&	 6.9138	&	 3.4569	&	 0.3734	\\
prop17	&	 7.1707	&	 7.1707	&	 3.5853	&	 0.4017	\\
prop18	&	 7.4371	&	 7.4371	&	 3.7185	&	 0.4321	\\
prop19	&	 7.7134	&	 7.7134	&	 3.8567	&	 0.4648	\\
prop20	&	 8.0000	&	 8.0000	&	 4.0000	&	 0.5000	\\
   \hline
        \end{tabular}
            \caption{Adaptive CMTM: Final $\sigma_{k,j}$ for each coordinate and each proposal used.}
\label{tab: acmtm_4dimmixnorm_sigmaj}

\end{table}

\begin{table}[H]
\setlength{\tabcolsep}{5pt}
    \begin{subtable}{.48\textwidth}
    \caption{Non-adaptive CMTM}
\label{tab: nonacmtm_4dimmixnorm_perf}
\centering
\begin{tabular}{rrrrrr}
  \hline
& Min. & Median & Mean & Max.  \\ 
  \hline
cputime(s)& 10.25 & 10.41 & 10.43 & 11.22\\
sq. jump &6.20 & 6.62 & 6.62 & 7.07 \\ 
\hline
\hline
&coord1 & coord2 & coord3 & coord4 \\
\hline
ACT & 41.96 & 41.25 &  1.64 &  1.64 \\
   \hline
   \end{tabular}
    \end{subtable}%
    \hspace{0.04\textwidth}
    \begin{subtable}{.48\textwidth}
\caption{Adaptive CMTM}
\label{tab: acmtm_4dimmixnorm_perf}
\centering
\begin{tabular}{rrrrrr}
  \hline
 & Min. & Median & Mean & Max.\\ 
  \hline
cputime(s)& 10.42 & 10.57 & 10.65 & 13.14 \\ 
sq. jump & 8.88 & 10.15 & 10.04 & 10.76 \\ 
\hline
\hline
&coord1 & coord2 & coord3 & coord4 \\
\hline
ACT &22.55 &22.46  &1.43  &1.00\\
   \hline
        \end{tabular}
    \end{subtable} 
     \caption{Comparison of performance indicators that were computed from 100 independently replicated runs. The tables contain statistics about the execution time for a complete run (cputime), the average squared jump distance and the ACT.}
         \label{tab: 4dimmixnorm_perf_acmtm}

\end{table}

\begin{table}[H]
   \begin{center}
\begin{tabular}{r|rrrrr}
  \hline
 & coord1 & coord2 & coord3 & coord4 \\  
  \hline
prop1	&	 0.04	&	 0.05	&	 0.05	&	 0.04	\\
prop2	&	 0.05	&	 0.05	&	 0.05	&	 0.05	\\
prop3	&	 0.05	&	 0.05	&	 0.05	&	 0.05	\\
prop4	&	 0.05	&	 0.04	&	 0.05	&	 0.05	\\
prop5	&	 0.05	&	 0.05	&	 0.05	&	 0.05	\\
prop6	&	 0.05	&	 0.05	&	 0.05	&	 0.05	\\
prop7	&	 0.05	&	 0.05	&	 0.05	&	 0.04	\\
prop8	&	 0.05	&	 0.05	&	 0.05	&	 0.05	\\
prop9	&	 0.05	&	 0.05	&	 0.05	&	 0.06	\\
prop10	&	 0.05	&	 0.05	&	 0.05	&	 0.05	\\
prop11	&	 0.05	&	 0.05	&	 0.04	&	 0.05	\\
prop12	&	 0.05	&	 0.05	&	 0.05	&	 0.05	\\
prop13	&	 0.05	&	 0.05	&	 0.05	&	 0.06	\\
prop14	&	 0.05	&	 0.05	&	 0.05	&	 0.05	\\
prop15	&	 0.05	&	 0.05	&	 0.05	&	 0.05	\\
prop16	&	 0.05	&	 0.05	&	 0.05	&	 0.05	\\
prop17	&	 0.05	&	 0.05	&	 0.05	&	 0.05	\\
prop18	&	 0.05	&	 0.05	&	 0.05	&	 0.04	\\
prop19	&	 0.05	&	 0.05	&	 0.05	&	 0.04	\\
prop20	&	 0.05	&	 0.05	&	 0.05	&	 0.04	\\
   \hline
 \end{tabular}
\caption{Adaptive CMTM: Rate of selection for each proposal and each coordinate.}
 \label{tab: acmtm_4dimmixnorm_pselect}
\end{center}
\end{table}

When  comparing  the rate of selection for each proposal, as reported  in Tables \ref{tab: mtm_4dimmixnorm_pselect} and \ref{tab: acmtm_4dimmixnorm_pselect},   we observe the almost constant selection probabilities for the ACMTM which suggests that all the proposal scales selected are important in the simulation. Finally, we also compare the acceptance frequencies for the selected proposals for CMTM and ACMTM, as shown in Tables  \ref{tab: mtm_4dimmixnorm_acrate} and  \ref{tab: acmtm_4dimmixnorm_acrate}, respectively. The adaptive version of the algorithm clearly  makes better use of the generated proposals. There are no longer any NA's, i.e. all  proposals are occasionally accepted in each coordinate. In fact, the acceptance rates for ACMTM are quite even, again suggesting a balanced use of the proposal distributions. In almost every instance the acceptance rates have gone up compared to the CMTM values in Table  \ref{tab: mtm_4dimmixnorm_acrate}.

    \begin{table}[H]
\centering
\begin{tabular}{r|rrrrrr}
  \hline
%  &\multicolumn{5}{c}{\(\sigma_{k,j}\)}\\
%\hhline{~-----}
 & coord1 & coord2 & coord3 & coord4 \\  
  \hline
prop1	&	 0.58	&	 0.66	&	 0.49	&	 0.60	\\
prop2	&	 0.57	&	 0.58	&	 0.58	&	 0.60	\\
prop3	&	 0.60	&	 0.65	&	 0.62	&	 0.60	\\
prop4	&	 0.63	&	 0.55	&	 0.59	&	 0.60	\\
prop5	&	 0.61	&	 0.59	&	 0.58	&	 0.65	\\
prop6	&	 0.65	&	 0.53	&	 0.60	&	 0.60	\\
prop7	&	 0.59	&	 0.59	&	 0.60	&	 0.62	\\
prop8	&	 0.64	&	 0.65	&	 0.58	&	 0.60	\\
prop9	&	 0.58	&	 0.57	&	 0.59	&	 0.60	\\
prop10	&	 0.57	&	 0.61	&	 0.60	&	 0.56	\\
prop11	&	 0.61	&	 0.66	&	 0.59	&	 0.54	\\
prop12	&	 0.57	&	 0.54	&	 0.62	&	 0.66	\\
prop13	&	 0.53	&	 0.54	&	 0.66	&	 0.60	\\
prop14	&	 0.55	&	 0.58	&	 0.57	&	 0.61	\\
prop15	&	 0.61	&	 0.60	&	 0.58	&	 0.55	\\
prop16	&	 0.58	&	 0.61	&	 0.60	&	 0.60	\\
prop17	&	 0.54	&	 0.65	&	 0.61	&	 0.57	\\
prop18	&	 0.58	&	 0.61	&	 0.58	&	 0.53	\\
prop19	&	 0.56	&	 0.56	&	 0.62	&	 0.60	\\
prop20	&	 0.61	&	 0.63	&	 0.66	&	 0.59	\\
   \hline
        \end{tabular}
        \caption{ACMTM:  Post-selection acceptance probabilities  for each proposal.}
        \label{tab: acmtm_4dimmixnorm_acrate}
\end{table}

\subsection{Convergence of Adaptive CMTM}
\label{sec: convergence}

We prove below the convergence of the adaptive CMTM algorithm described in Section \ref{sec:algomtm}. As explained in Section \ref{sec:algomtm}, Diminishing Adaptation condition holds by the construction of the adaption mechanism. 

\begin{theorem}
\label{convthm}
Consider the adaptive CMTM algorithm in Section \ref{sec:algomtm} to sample from state space $\mathcal{X}$ that is an open subset of $\mathbf{R}^d$ for some $d \in \mathbf{N}$. Let $\pi$ be a target probability distribution, which has a continuous positive density on $K$ with respect to the Lebesgue measure. Then, the adaptive CMTM algorithm converges to stationarity as in 
\begin{align}
\label{eq:ttlvard}
\lim_{n\to \infty} \sup_{A \in \mathcal{F}}
|\mathbf{P}(X_n \in A) - \pi (A) | = 0.
\end{align}
\end{theorem}

\begin{proof}

By \cite{roberts2007coupling}, the convergence of an adaptive
MCMC algorithm as in \eqref{eq:ttlvard} can be ensured by two
conditions Diminishing Adaptation and Containment. Our algorithm
satisfies Diminishing Adaptation (DA) as explained in Section
\ref{sec:algomtm}. So, it suffices to show that
our algorithm satisfies the Containment condition.

The Containment condition of \cite{roberts2007coupling}
(see also \cite{craiu2014stability, rosenthal2016ergodicity} states that
the process's convergence times are bounded in probability, i.e.\ that
$\{M_{\epsilon} (X_n, \Gamma_n)\}_{n=1}^\infty$
is bounded in probability,
where \(M_{\epsilon} (x, \gamma )
:= \inf \{n \geq 1 : \| P^{n}_{\gamma} (x, \cdot) - \pi (\cdot) \|
\leq \epsilon \}\) for all \(\epsilon > 0\), and $P_\gamma^n$ is a
fixed $n$-step proposal kernel.  

We proceed similarly to the proof of Proposition~23 of
\cite{craiu2014stability}. By our assumption~(A1), the process $\{X_n\}$
is bounded in probability, in fact $\|X_n\| \le L$ for all~$n$.
To continue, we let
$\Y$ be the collection of all $d \times m$ matrices of
real numbers in \([\epsilon, L]\).  Then by our assumption (A2),
$\Y$ is compact.  Here each $\gamma\in\Y$ corresponds to
a particular choice of MTM proposals, where
$\gamma_{k,j}$ equals the scaling of the
$j$th proposal kernel for the $k${th} coordinate.
And, our adaption rule is such that choosing which
$\gamma\in\Y$ to use for each iteration \(n\) is determined
by the past and/or current information obtained from the chain.

Next, let $P_\gamma$ be the Markov kernel corresponding to
one full sequence of updates for all coordinates of the chain, in
sequence.  Then $P_\gamma$ is
Harris ergodic to $\pi$, since it is known that
any {\it non}-adaptive CMTM algorithm
must converge to \(\pi\)
(cf.\ \cite{liu2000multiple, casarin2013interacting}).
It follows that
$\lim_{n \to \infty} \Delta (x, \gamma, n):= \| P^{n}_{\gamma}
(x, \cdot) - \pi (\cdot) \| = 0$ for each $(x,\gamma)$,
where $\|\cdots\|$ is the usual total variation distance convergence metric.
Now, with our
algorithm as set up in Section \ref{sec:algomtm}, $\Delta (x, \gamma,
n)$ is a continuous function of $(x,\gamma)$: indeed, it is a
composition of single-coordinate MTM updates each of which is
continuous as in the proof of Corollary~11 of \cite{roberts2007coupling}.

To finish, we note (following \cite{rosenthal2016ergodicity})
that by Dini's Theorem, 
\begin{align*}
\lim_{n \to \infty}\sup_{x\in \mathcal{C}} \sup_{\gamma \in
\mathcal{Y}} \Delta (x, \gamma, n) = 0
\end{align*}
for any compact set \(\mathcal{C} \subset \mathcal{X}\). Hence,
for any \(\epsilon >0 \), there is \(D<\infty\) such that \(\sup_{x\in
\mathcal{C}} \sup_{\gamma \in \mathcal{Y}} \Delta (x, \gamma,
D) < \epsilon \). It follows that \(\sup_{x\in \mathcal{C}} \sup_{\gamma
\in \mathcal{Y}} M_{\epsilon} (x, \gamma ) \le D < \infty \).
In particular, choosing $C=K$ from our assumption (A1),
we know that
$P (X_n \not \in K)=0$ for all \(n\), so if \( D \ :=
\sup_{x\in \mathcal{K}} \sup_{\gamma \in \mathcal{Y}} M_{\epsilon}
(x, \gamma )\), then for any $\delta>0$,
\(P (M_{\epsilon} (X_n, \Gamma_n) > D ) = 0 \leq
\delta \) for all \(n\). In particular,
$\{M_{\epsilon} (X_n, \Gamma_n)\}_{n=1}^\infty$ is bounded in
probability.
Therefore, the Containment condition holds, thus finishing the proof.
\end{proof}

% and $\Delta (x, \gamma, n)$ is a non-increasing function of $n$.

% i.e.~\eqref{eq:containment} is satisfied.

\medskip
\begin{remark}
Our theorem is still valid if
the number of proposals $m$ is allowed to change from iteration to iteration,
provided $m$ is forced to remain between
$1$ and some large finite upper bound $M$.
Indeed, in that case $\Y$ is a discrete union of $M$ different
collections of $d \times m$ matrices, and $\Delta(x,\gamma,n)$ is
continuous separately on each collection, and the rest of the proof can
then proceed without further change.
\end{remark}

\section{Applications}
\label{sec: appl}

In the following examples we compare the CMTM and AMCTM started with the same set of  $\sigma_{k,j}$. We also compare their performance with CMH and adaptive CMH. The  design of the latter is based on the theoretical results  of  \cite{gelman1996efficient} and \cite{roberts2001optimal} who found that the optimal acceptance rate for one-dimensional Metropolis algorithm is 0.44  and therefore adjusts the proposal variance to  get an acceptance rate close to this value  for each coordinate.

First we compare CMTM ( with different number of proposals $m$) with CMH, both with generic proposals.
For CMTM with  $m$   proposals we set $\sigma_{k,j}=2^{j-1-\lfloor m/2 \rfloor}$ for each coordinate $1\le j \le m$. The CMH's proposals are fixed at 1 for each coordinate.

In second comparisons we compare adaptive CMTM with different number of proposals  and adaptive CMH. The starting $\sigma$'s are identical to the ones used in their non-adaptive counterparts.

For all the examples we use  the effective sample size (ESS) and ESS/CPUtime (CPUtime is the time needed to complete the simulation)  to compare the efficiency of MCMC algorithms. The latter is particularly relevant for algorithm comparison since it is a way to quantify the resource allocation efficiency.  Since  ESS $=w/\tau$, where $w$ is the number of samples obtained from a Markov chain and $\tau$ is the ACT, one can see that ESS is equivalent to ACT. One may intuitively interpret ESS the number of iid samples from the target that would contain the same amount of information  about the target as the MCMC sample. The first half of the chains'  10000 realizations   is discarded and the remaining samples are used to calculate the ACT. The reported ESS is based on  averaging the  ACT over 50 independent runs.

\subsection{Variance Components Model}
\label{sec: vcm}

The Variance Components Model (VCM) is a typical hierarchical model, often used in Bayesian statistics community. Here, we use the data on batch to batch variation in dyestuff yields. The data were introduced in \cite{davies1967statistical} and later analyzed by \cite{box1973bayesian}. The Bayesian set-up of the Variance Components Model on dyestuff yields is also well-described in \cite{roberts2004general}. The data records yields on dyestuff of 5 samples, from each of 6 randomly chosen batches. The data is shown in Table~\ref{table:dyes}.
\begin{table}[H]
\centering
\begin{tabular}{r|rrrrr}
  \hline
Batch 1 &1545 &1440 &1440 &1520 &1580\\
Batch 2 &1540 &1555 &1490 &1560 &1495\\
Batch 3 &1595 &1550 &1605 &1510 &1560\\
Batch 4 &1445 &1440 &1595 &1465 &1545\\
Batch 5 &1595 &1630 &1515 &1635 &1625\\
Batch 6 &1520 &1455 &1450 &1480 &1445\\
  \hline
  \end{tabular}
  \caption{Dyestuff Batch Yield (in grams)}
  \label{table:dyes}
  \end{table}

Let $y_{ij}$ be the yield on the dyestuff batch, with $i$ indicating which batch it is from and $j$ indexing each individual sample from the batch. The Bayesian model is then constructed as:
\begin{align*}
y_{ij}|\theta_i, \sigma_e^2 \sim N(\theta_i, \sigma_e^2), \qquad i = 1,2,..., K, \quad j = 1,2,...,J
\end{align*}
where \(\theta_i|\mu,  \sigma_{\theta}^2 \sim N(\mu,
\sigma_{\theta}^2)\). \(\theta_i\)'s are
conditionally independent of each other given $\mu,  \sigma_{\theta}^2$. The priors for the $\sigma_{\theta}^2, \sigma_e^2$ and $\mu$ are: \(\sigma_{\theta}^2 \sim IG(a_1, b_1)\), \(\sigma_e^2 \sim IG(a_2, b_2)\) and \(\mu \sim N(\mu_0, \sigma_0^2)\). Thus, the posterior density function of this VCM model is 
\begin{align*}
&f(\sigma_{\theta}^2, \sigma_{e}^2, \mu, \theta_i | y_{ij}, a_1, a_2, b_1, b_2, \sigma_0^2 )
 \propto \\
&(\sigma_{\theta}^2)^{-(a_1+1)} e^{-b_1/\sigma_{\theta}^2} (\sigma_{e}^2)^{-(a_2+1)} e^{-b_2/\sigma_{e}^2} e^{-(\mu-\mu_0)^2/2\sigma_0^2}
\prod_{i=1}^K \frac{e^{(\theta_i-\mu)^2/2\sigma^2_{\theta}}}{\sigma_{\theta}} \prod_{i=1}^K \prod_{j=1}^J \frac{e^{(y_{ij}-\theta_i)^2/2\sigma^2_{e}}}{\sigma_{e}} 
\end{align*}

We set  the hyperparameters \(a_1 = a_2 = 300\) and \(b_1 = b_2 = 1000\),  making inverse gamma priors very concentrated. We also set \(\sigma_0^2 = 10^{10}\).

Figure \ref{fig:vcncomc_ess_adp} shows ESS and ESS/CPU (averaged over 50 runs) of the CMTM algorithms with and without adaption and of standard CMH and adaptive CMH algorithm. For both CMTM algorithms (with and without adaption), the starting proposals were generic for every coordinate as described above.

\begin{figure}[H]
\includegraphics[height=0.33\textheight]{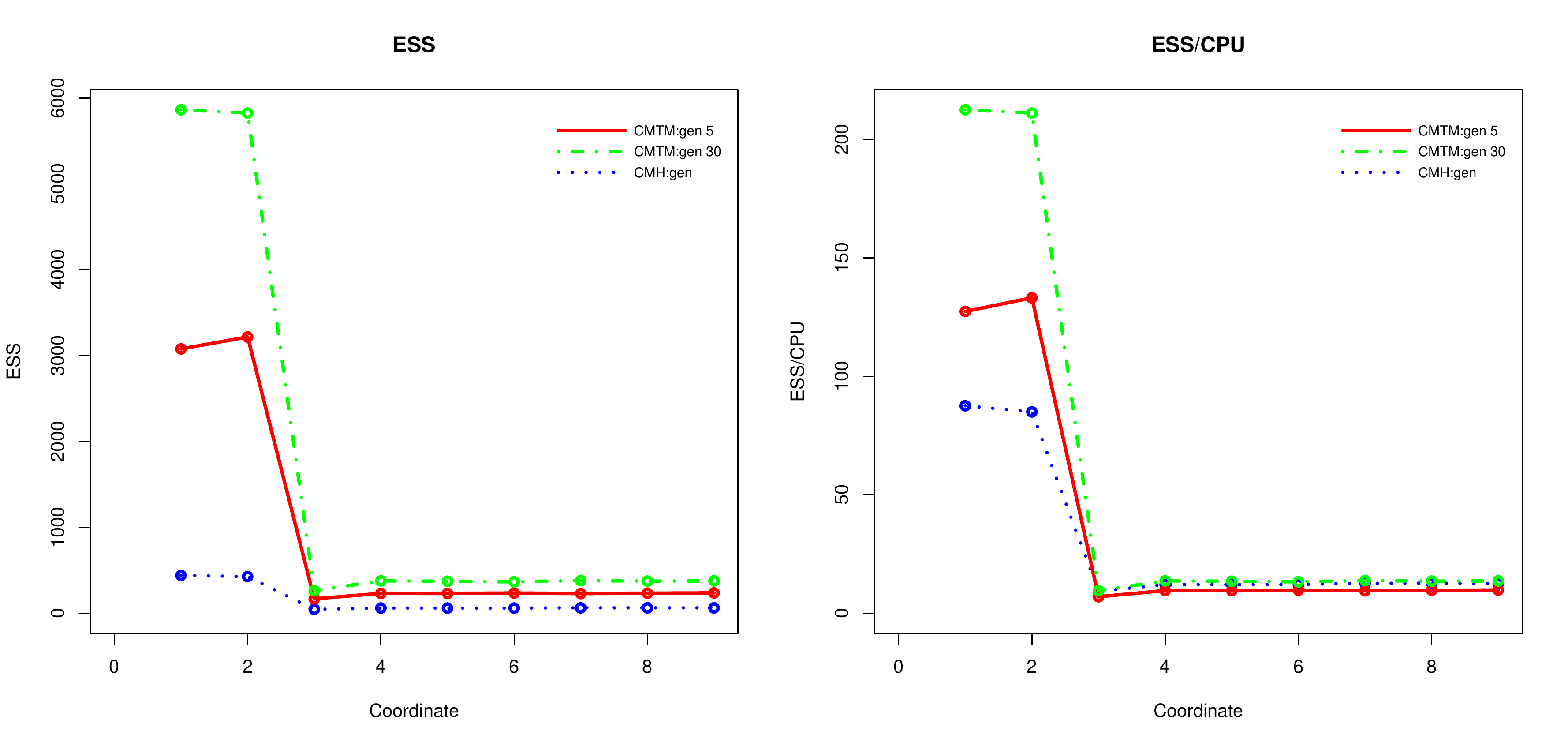}

\includegraphics[height=0.33\textheight]{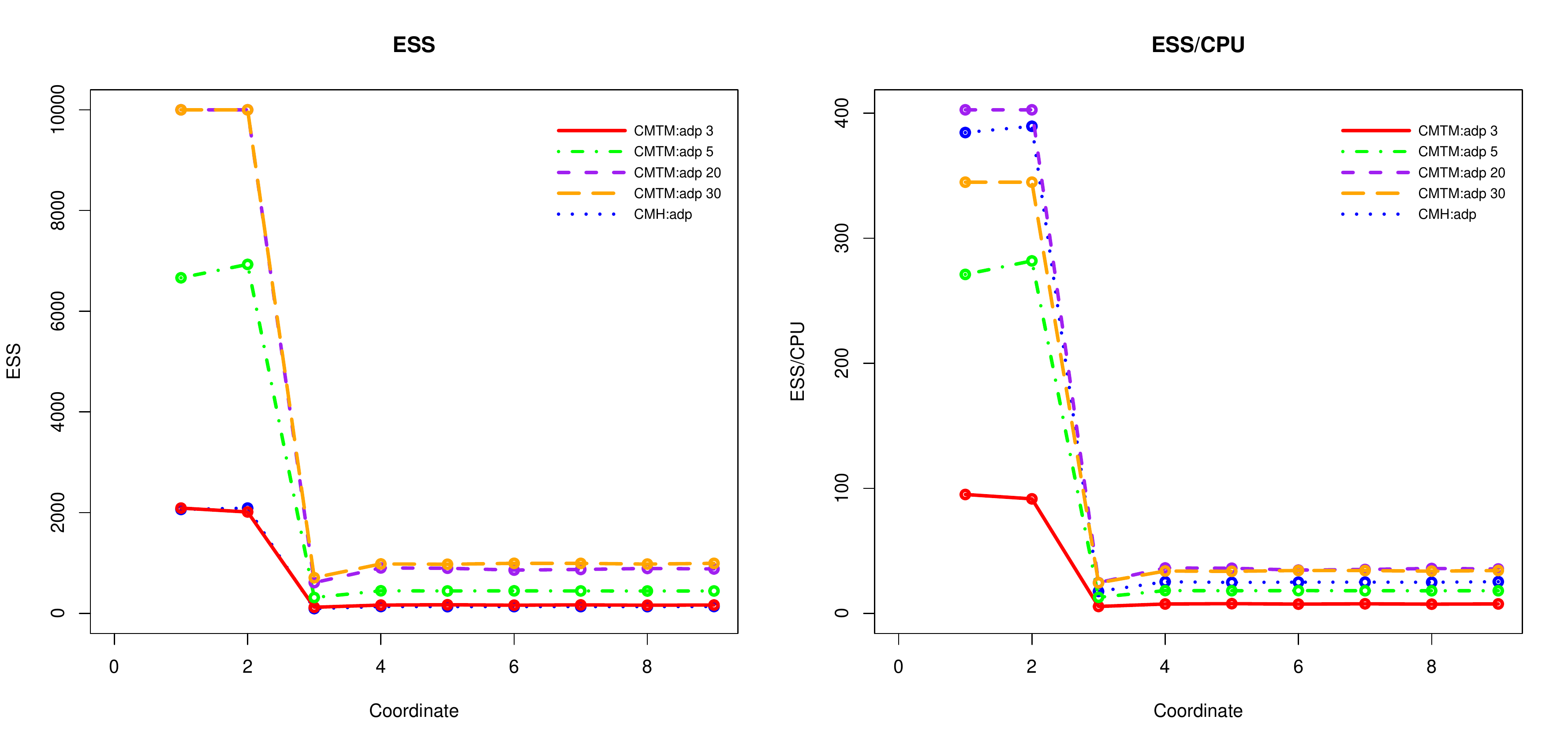}

\caption{Variance components model. {\bf Top Row:} Comparison between non-adaptive CMH and  CMTM with 5 and 30 generic proposals.  The red and green lines show the ESS (left panel) or ESS/CPU (right panel) for  the CMTM with 5 and 30 proposals, respectively, the blue line presents the same for CMH. {\bf Bottom Row}: Comparison between ACMTM with 3, 5, 20 or 30 proposals and the adaptive CMH. The red, green, purple and orange lines show ESS (left panel) or ESS/CPU (right panel) for ACMTM with 3, 5, 20 and 30 proposals, respectively, and the blue shows the performance for the adaptive CMH.}
\label{fig:vcncomc_ess_adp}
\end{figure}

The plots for non-adaptive samplers clearly show that CMTM with 30 proposals is the most efficient in ESS and even when CPU time is taken into account
it still performs better than CMH. Similar results is evident for adaptive samplers. Clearly adaptive CMTM with 20 or 30 proposal have much
better ESS than adaptive CMH. When CPU time is considered than adaptive CMTM with 20 proposals is the most efficient.

\subsection{``Banana-shaped'' Distribution}
\label{sec: banana}

The ``Banana-shaped'' distribution was originally presented in \cite{haario1999adaptive} as an irregularly-shaped target that may call for different proposal distributions for the different parts of the state space. 

The target density function of the ``banana-shaped'' distribution is constructed as $f_B = f \circ \phi_B$, where $f$ is the density of $d-$dimensional multivariate normal distribution $N(\mathbf{0}, \text{diag}(100,1,1,\ldots,1))$ and $\phi_B (\mathbf{x}) = (x_1, x_2 + Bx_1^2 - 100B, x_3, \ldots, x_d)$. $B>0$ is the nonlinearity parameter and the non-linearity or ``bananacity'' of the target distribution increases with $B$. The target density function is 

\begin{align*}
 f_B(x_1, x_2,\ldots,x_d) \propto \exp  [ - x_1^2/200 - \frac{1}{2} (x_2 + B x_1^2 - 100 B)^2
  - \frac{1}{2}(x_3^2 + x_4^2 + \ldots + x_d^2)  ]. 
\end{align*}

We set \(B = 0.01\) and \(d = 10\),                       the results are shown in Figure \ref{fig:banana_ess_adp} (averaged over 50 runs starting with generic proposals).  
\begin{figure}[H]
\includegraphics[height=0.33\textheight]{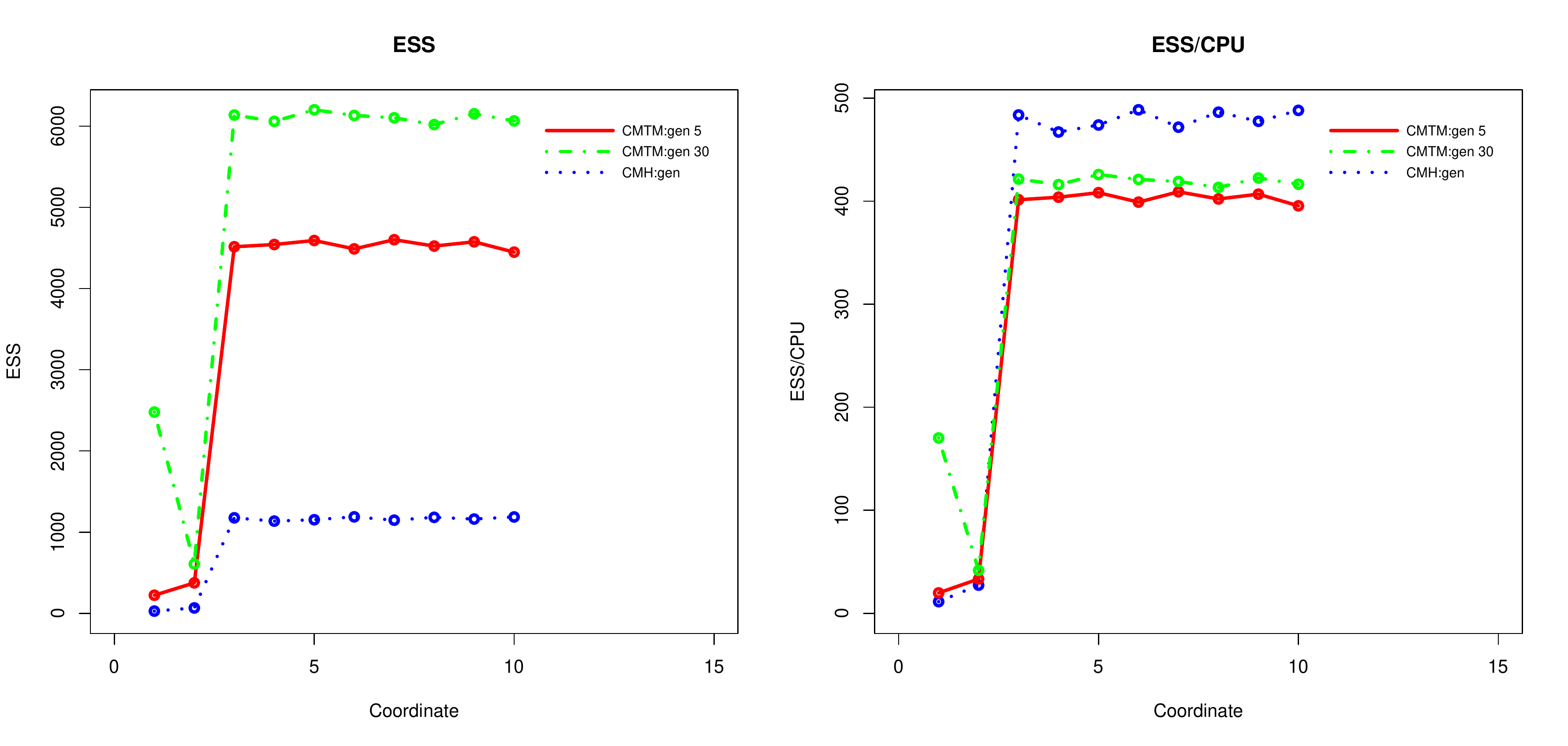}
\includegraphics[height=0.33\textheight]{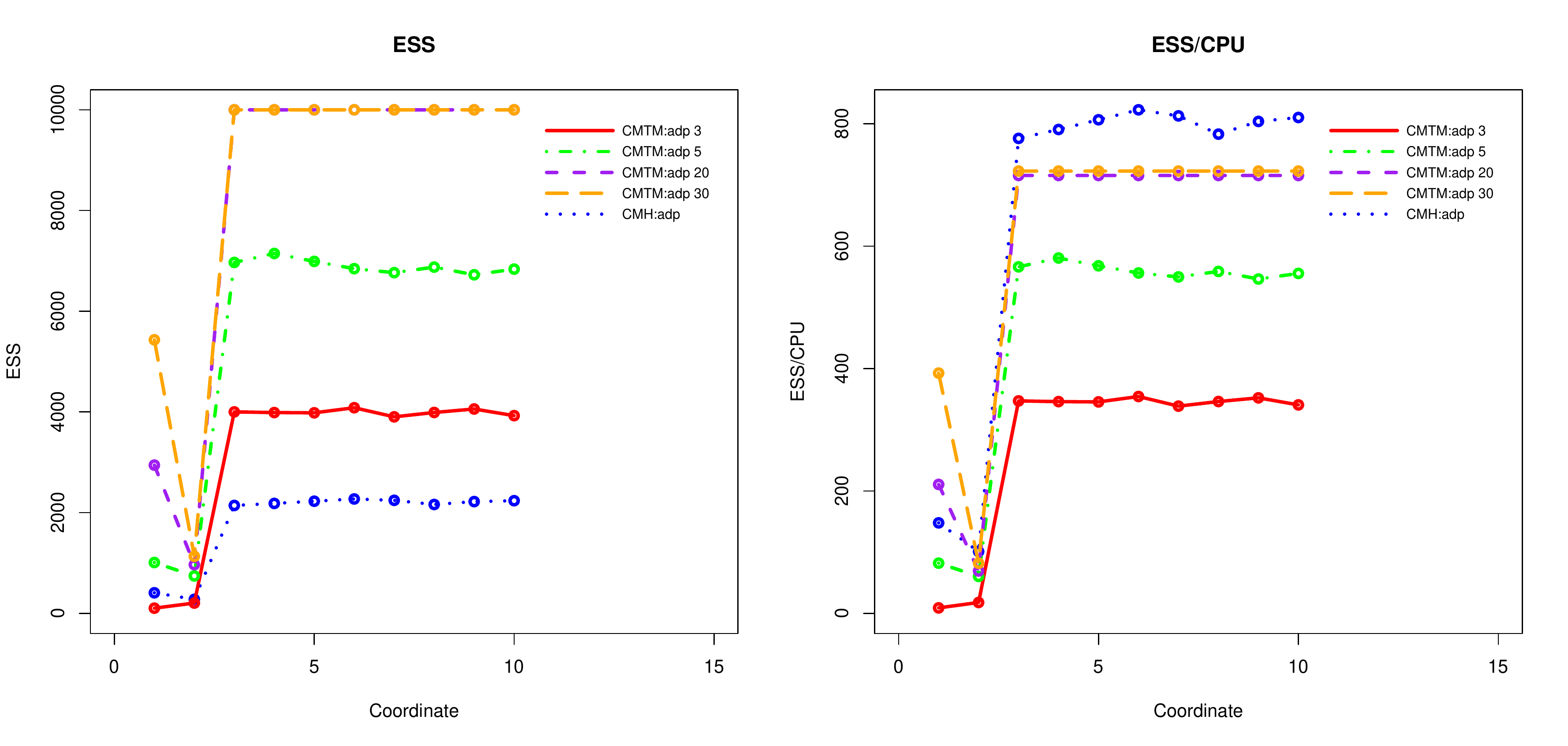}
\caption{Banana-shaped distribution. {\bf Top Row:} Comparison between non-adaptive CMH and  CMTM with 5 and 30 generic proposals.  The red and green lines show the ESS (left panel) or ESS/CPU (right panel) for  the CMTM with 5 and 30 proposals, respectively, the blue line presents the same for CMH. {\bf Bottom Row}: Comparison between ACMTM with 3, 5, 20 or 30 proposals and the adaptive CMH. The red, green, purple and orange lines show ESS (left panel) or ESS/CPU (right panel) for ACMTM with 3, 5, 20 and 30 proposals, respectively, and the blue shows the performance for the adaptive CMH.}
\label{fig:banana_ess_adp}
\end{figure}

Focusing on ESS plots, CMTM and adaptive CMTM with 30 proposals clearly outperform standard CMH and adaptive CMH in all coordinates. When
CPU time is taken into account then CMH and adaptive CMH performs a little better than CMTM algorithms on most coordinates.
However on coordinate 1, CMTM methods perform much better than CMHs, actually by a factor of 2.5 or more.

\subsection{Mixture of 20-dimensional Gaussians}
\label{sec: mixture}

\begin{figure}[H]
\includegraphics[height=0.33\textheight]{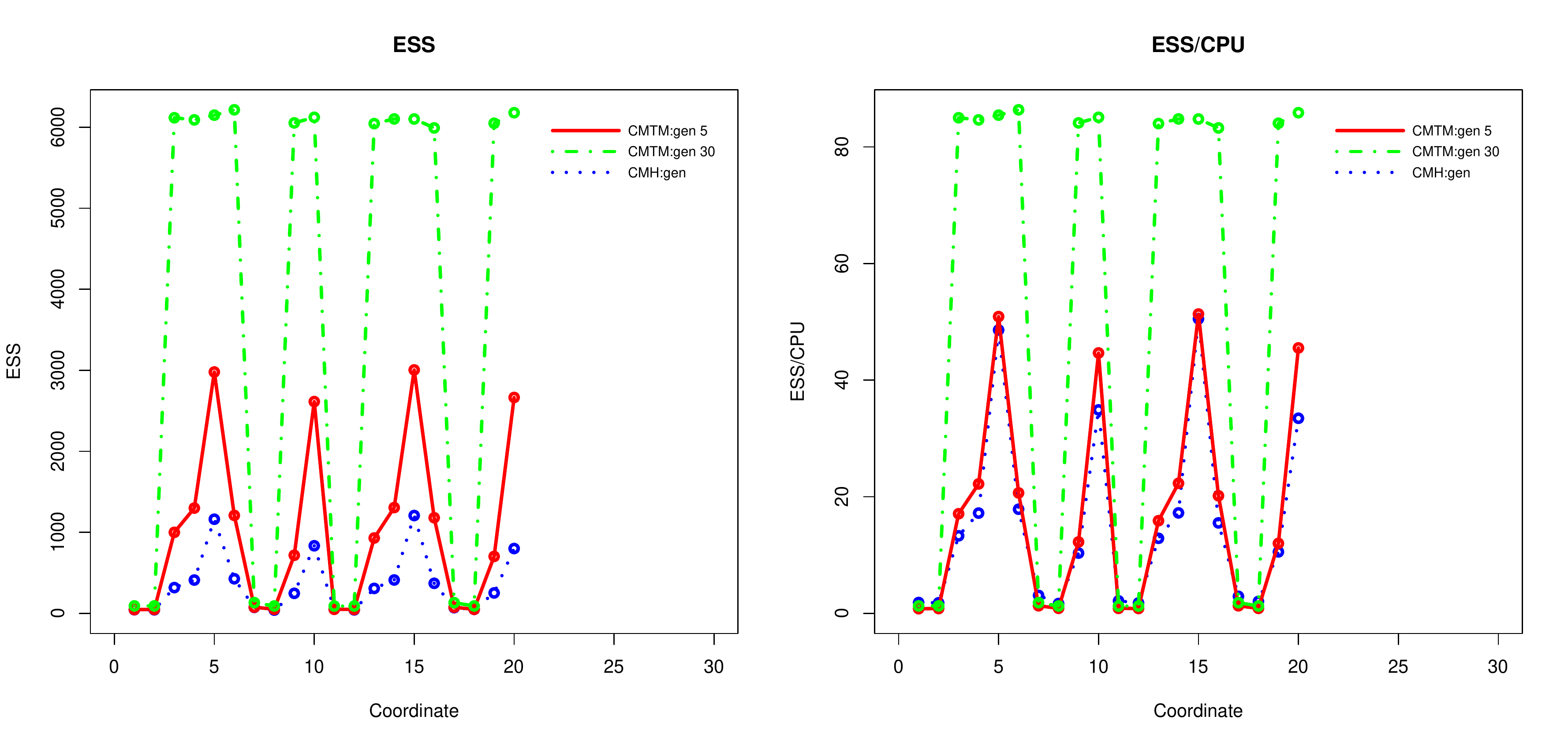}
\includegraphics[height=0.33\textheight]{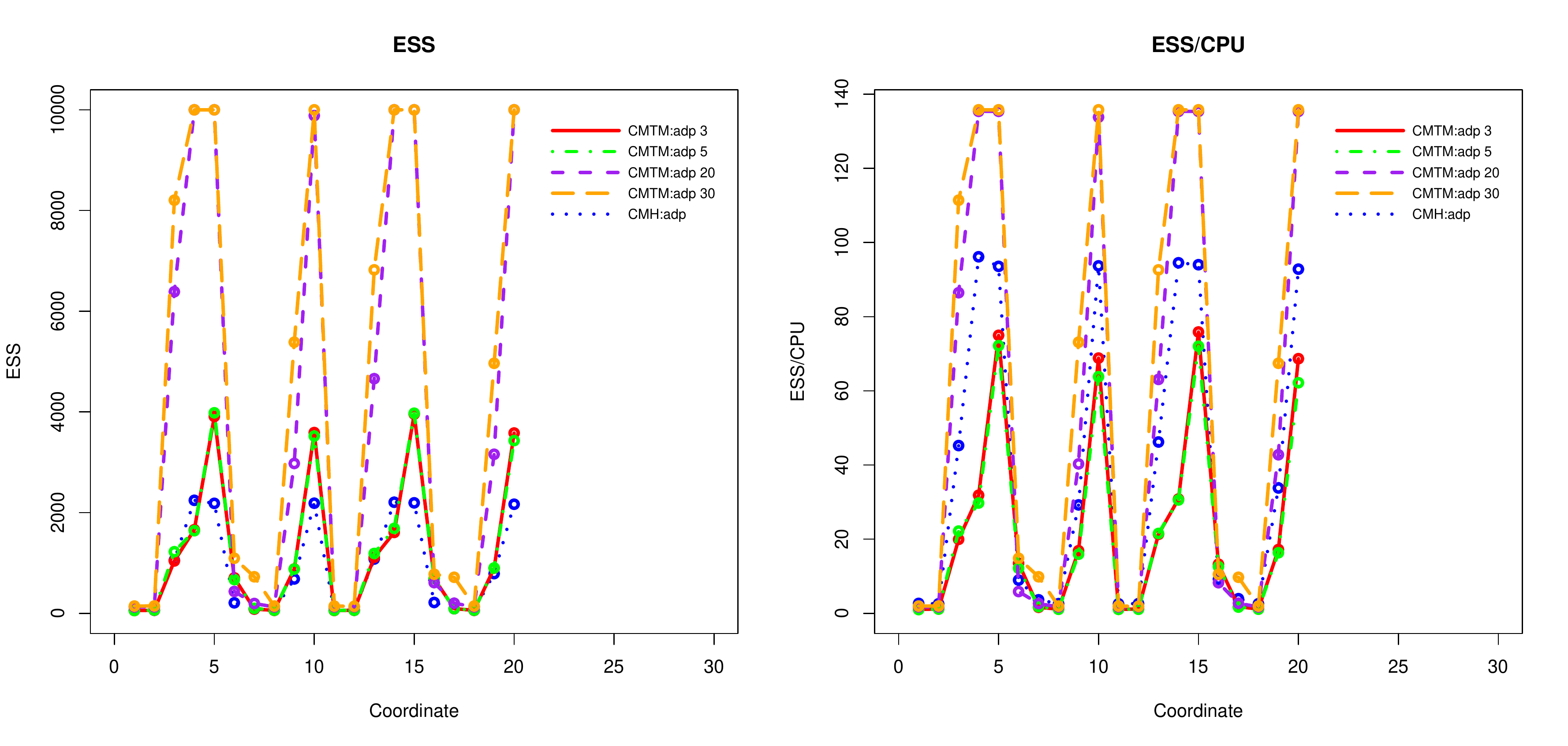}
\caption{20-dimensional mixture distribution. {\bf Top Row:} Comparison between non-adaptive CMH and  CMTM with 5 and 30 generic proposals.  The red and green lines show the ESS (left panel) or ESS/CPU (right panel) for  the CMTM with 5 and 30 proposals, respectively, the blue line presents the same for CMH. {\bf Bottom Row}: Comparison between ACMTM with 3, 5, 20 or 30 proposals and the adaptive CMH. The red, green, purple and orange lines show ESS (left panel) or ESS/CPU (right panel) for ACMTM with 3, 5, 20 and 30 proposals, respectively, and the blue shows the performance for the adaptive CMH.}
\label{fig:mixt_ess_adp}
\end{figure}

We are also examining the gains brought by the ACMTM in the case of multimodal distributions. We consider the mixture 
$$0.5N_20(\mu_1,\Sigma_1) + 0.5 N_20(\mu_2,\Sigma_2)$$
where
\begin{align*}
\mu_1 = &(5,  5 , 0, 0  ,0 , 0, 10 ,15 , 0 , 0  ,5 , 5,  0 , 0 , 0 , 0 ,10, 15 , 0,  0),\\
\mu_2 = &(10 ,10 , 0 , 0 , 0 , 0 , 7 ,20 , 0 , 0 ,10, 10 , 0 , 0 , 0 , 0 , 7 , 20 , 0 , 0),\\
\Sigma_1  = \mbox{diag}&( 16.00 ,16.00 , 0.25 , 4.00 ,1.00  ,0.01 ,9.00 ,16.00 , 9.00 ,\\
 &0.01 ,16.00, 16.00, 0.25  ,4.00,  1.00 , 0.01 , 9.00 ,16.00  ,9.00  ,0.01),\\
\Sigma_2 = \mbox{diag}&( 16.00 ,16.00 , 6.25 ,4.00 , 1.00 , 4.41 , 9.00  ,16.00 ,0.25  ,\\
& 0.01, 16.00 ,16.00 ,6.25 , 4.00 , 1.00 , 4.41 , 9.00, 16.00 , 0.25 , 0.01).\\
\end{align*}

 In this example, CMTM methods with 30 proposals (in each coordinate) is the most efficient in ESS and ESS/CPU. The comparison is reported in Figure \ref{fig:mixt_ess_adp}. We note that the 
adaptive and non adaptive versions of CMTM perform much better than the CMHs counterparts.

The ESS/CPU calculations  suggest   that the best performance is achieved when the  number of chains $m$ is between 20 and 30. When programming the examples (the programs are available as online supplemental material), we have taken advantage of the software R's ability to handle vectorial operations much more efficiently than loops. When similar savings  can be  obtained, we recommend using $m=20$ in practice. In instances where the likelihood is expensive to compute due to the large number of observations in the data, embarrassingly parallel strategies could be used efficiently in conjunction with ACMTM \citep{emb,cons,weierstrass,reih}.

It is also important to note that in all 3 examples described above  adaptive CMTM is always more efficient than CMTM with generic proposals. CPU time for both are about the same but ESS generally much larger for the latter. Hence adaptive CMTM generally produces much better results and it is
advisable to use it for real-world problems especially since it only requires a few lines of extra code.

\section{Conclusion and Discussion}

It is known that adaptive algorithms can be highly influenced by initial values given to their simulation parameters and by the quality of the chain during initialization period, i.e. the period during which  no modifications of the transition kernel take place. ACMTM is no exception, but some of its features can be thought of as means towards a more robust behaviour. For instance, the fact that we can start with multiple proposals makes it less likely that all initial values will be poor choices for a given coordinate.  The ACMTM is motivated by situations in which the sampler requires very different proposals across coordinates and across regions of the state space. In such situations, traditional adaptive samplers are known to fail unless special modifications are implemented \citep{cry, cbd}, but even these tend to underperform when $d$ is high. 

 The adaption mechanism is very rapid as the scales can change in multiple of 2's and is also stable since modifications to the kernel occur only if over selection from one of the boundary scale proposals is detected. Thus, even if proposal scales are not perfect but good enough, they would not  change much under this adaptive design.  
 
The increase in CPU time is the price we pay for the added flexibility of having multiple proposals and the ability to dynamically choose the ones that fit the region of the space so that acceptance rate and mixing rates are improved. And while this tends to attenuate the ACMTM's  efficiency, one cannot find among the algorithms we used for comparison in this paper  one that is performing better {\it on average} even after taking CPU time into account.  However, we recommend using ACMTM in difficult sampling problems (e.g. multimodal target, variable variances for the conditional distributions across the sample space) when other  approaches do not perform well. 

Finally, it is the authors belief that AMCMC samplers will  be  used in practice more if their motivation is intuitive and their implementation is easy enough. We believe that the ACMTM fulfills these basic criteria and further modifications can be easily implemented once  new needs are identified.

\section*{Acknowledgement}

Funding support for this work was provided by individual grants to RC and JSR from the Natural Sciences and Engineering Research Council of Canada.

\clearpage

\bibliographystyle{jasa}
%\bibliography{mybib}

\end{document}